\documentclass[12 pt]{amsart}
\usepackage{amssymb,amsfonts,amsthm, amsmath, color}
\usepackage{graphicx}
\usepackage{float}
\usepackage{bm}
\usepackage[normalem]{ulem}
\usepackage{xcolor}

\newtheorem{thm}{Theorem}[section]

\newtheorem{lem}[thm]{Lemma}
\newtheorem{prop}[thm]{Proposition}

\theoremstyle{definition}

\theoremstyle{remark}
\newtheorem{rem}[thm]{Remark}
\numberwithin{equation}{section}

\begin{document}
\newcommand\redsout{\bgroup\markoverwith{\textcolor{red}{\rule[0.5ex]{2pt}{3pt}}}\ULon}
\title[Information theoretic measures of structured neural
networks]{Entropy, mutual information, and systematic measures of
  structured spiking neural networks}
\author{Wenjie Li}
\address{Wenjie Li: Department of Mathematics and Statistics, Washington University, St. Louis, MO, 63130, USA}
\email{li.wenjie@wustl.edu}
\author{Yao Li}
\address{Yao Li: Department of Mathematics and Statistics, University
  of Massachusetts Amherst, Amherst, MA, 01002, USA}
\email{yaoli@math.umass.edu}
\thanks{Yao Li is partially supported by NSF DMS-1813246.}


\begin{abstract}
 The aim of this paper is to investigate various information-theoretic
 measures, including entropy, mutual information, and some systematic
 measures that based on mutual information, for a class of structured spiking neuronal network. In order
 to analyze and compute these information-theoretic measures for large
 networks, we coarse-grained the data by ignoring the order of spikes that fall into the same small
 time bin. The resultant coarse-grained entropy mainly capture the
 information contained in the rhythm produced by a local population of
 the network. We first proved that these information theoretical
 measures are well-defined and computable by proving the stochastic
 stability and the law of large numbers. Then we use three neuronal network examples, from simple
 to complex, to investigate these information-theoretic
 measures. Several analytical and computational results about
 properties of these information-theoretic
 measures are given.

\end{abstract}
\maketitle

\section{Introduction}
There has been a long history for researchers to use information
theoretic measures, such as entropy and mutual information, to study
activities of neurons \cite{strong1998entropy, nemenman2004entropy, borst1999information}. It is important to understand how
neuronal networks, including our brains, encode and decode
information. It is well known that neurons transmit information by
time series of spike trains. A common approach to estimate neuronal
entropy is to divide the time series of spike trains into a collection
of binary ``words''. More precisely, the time axis is divided into
many time windows with $m$ ``sub-windows''. A sub-window takes value
$1$ if a spike is recorded in it and $0$ otherwise. This gives a
binary ``word'' with $m$ ``letteres''. Entropy is then estimated
through frequencies of those words. 

However, estimating entropy becomes difficult for larger neuronal
network. If one considers the time serie generated by each neuron
separately, then one needs to consider all possible values of a large
vector of ``words'', which grows exponentially with the network size. If
one considers the time series of all spikes produced by the neuronal network, the
time window has to be extremely small to avoid two spikes fall into
the same bin. Either approach makes the sample size much smaller than the
number of possible configurations. This makes estimating entropy very
difficult, in spite of many results on estimations in the undersample
regime \cite{nemenman2004entropy, strong1998application,
  strong1998entropy}.

The first aim of this paper is to study information theoretic measures of a
structured neural network model introduced in \cite{li2019well} and
\cite{li2019stochastic}. Neurons in this network have integrate-and-fire
dynamics. Both the configuration of neurons and the rule of
interactions among neurons are simplified to make the model
mathematically and computationally tractable. It was shown in
\cite{li2019well} and
\cite{li2019stochastic} that this model is still able to produce
a rich dynamics of spiking patterns. In particular, this model can
produce multiple firing events (MFEs), which is a partially
synchronized spiking activities that have been observed in other more
realistic models \cite{rangan2013dynamics, rangan2013emergent,
  chariker2015emergent, chariker2016orientation}. The only difference is that
postsynaptic neurons in this paper are given by a {\it fixed}
connection graph rather than decided on-the-fly. 

To study information theoretical measures for large networks, we need
to consider the coarse-grained entropy instead. The idea is still to
divide time series of spikes into many time windows and construct ``words''. But different
from traditional approaches, here we do not distinguish the order of
spikes that fall into the same time window. This method has some
similarity to the multiscale entropy analysis used in many applications
\cite{costa2005multiscale}. In a large neuronal
network, the spike count in a time window can be large. Hence we use a
partition function to further reduce the total number of possible
``words''. In addition, we prove a law of large numbers of spike
counts, which says that the coarse-grained entropy defined in this
paper is both well-defined and computable.

The biological motivation of our definition is that MFEs produced by a
neuronal network can have fairly diverse spiking patterns (See Figure
\ref{fig5} and Figure \ref{fig6}).  In addition, it has
been argued that some information in the brain is indeed communicated through
resonance or synchronization of the Gamma oscillation
\cite{hahn2018portraits, hahn2014communication}, which is believed
to be modeled by MFEs in neuronal network models \cite{rangan2013emergent,
  chariker2015emergent}. This prompts us
to study information contained in those spiking volleys. By collecting spike
counts in time windows, we are able to obtain the uncertainty of a spiking
pattern. Heuristically, if a spiking pattern is completely
homogeneous, it contains little information from a coarse-grained
sense, as the spike count in a time window has little variation. Same
thing happens if the spiking activity is completely synchronized, at
which we have ``all-or-none'' spike counts in a time window. In
contrast, the spiking pattern contains more information if it consists
of MFEs, which are only partially synchronized, and the degree of synchronization has
high variation. This is confirmed by our numerical study. 

The definition of coarse-grained entropy can be extended to multiple
local populations. This gives the concept of mutual information, which
measures the amount of information shared by two local populations of
a neuronal network. To numerically study the mutual information, we
introduce three cortex models, from simple to complicated. The first
model only has two interconnected hypercolumns, with no geometry
structure. The second model aims to describe two layers of a piece of the 
cortex, each of which consists of many hypercolumns. We are interested
in the effect of feedforward and feedback connections expressed in information
theoretical measures. The third model is about layer 4 and layer 6 of
the primary visual cortex. In addition to spatial structures of
hypercolumns, there are orientation columns in both layers and long
range connections in layer 6. Control parameters are the magnitude of
feedforward connections, feedback connections, and long range
connections. Mutual information is also studied in this model. We find
that feedforward and feedback connections can enhance mutual
information between the two layers in all cases.

The second aim of this paper is to quantify some systematic measures,
such as degeneracy and complexity, to spiking neuronal networks. These
systematic measures are proposed in the study of systems biology
\cite{edelman2001degeneracy, tononi1999measures,
  whitacre2010degeneracy} and quantified for ODE-modeled networks in our earlier
work \cite{li2012quantification, li2016systematic}. Both degeneracy and complexity can be measured by a linear
combination of mutual information between components of a network. The
introduction of coarse-grained entropy makes these
systematic measures both well-defined and computable. 

Biologically speaking, the degeneracy measures the ability of structurally
different components
of a neuronal network to perform similar functions on a certain
target. And the (structural) complexity measures how different components in a neuronal
network functionally depend on each other. In this paper, degeneracy
and complexity are defined using coarse-grained entropy. We also
proved that a neuronal network with high degeneracy must be
(structurally) complex. Finally, the dependency of degeneracy and
complexity on certain network parameters is studied for our cortex
models.  

The organization of this paper is as follows: Section 2 defines our
structured spiking neuronal network, and three cortex models that are used
in later numerical studies. Section 3 defines the coarse-grained
entropy and proves the law of large numbers of spike counts, which
implies that the coarse-grained entropy is well-defined and
computable. Section 4 and 5 study mutual information and systematic
measures, respectively. Section 6 is the conclusion.

\section{Structured spiking neural network model}
It is well-known that the cerebral cortex has many
substructures. In particular, a functional organization called
cortical column or hypercolumn is believed to be the ``functional unit of
information processing''. Neurons in the same hypercolumn usually have
similar receptive fields. In the visual cortex, a hypercolumn can be
further divided into many orientation columns. Each orientation column
only responds to stimulations with a certain orientation. This
motivates us to propose a structured spiking neural network model that
has two scales at the level of individual neurons and hypercolumns,
respectively. 

\subsection{Network description}
We consider a large population of neurons that is divided into many
local structures (hypercolumns or orientational hypercolumns), called
local populations. The
following assumptions are made in order to describe the neuronal
activity of this population by a mathematically
tractable Markov process. Note that the first three assumptions are
identical to those in our earlier papers \cite{li2019well, li2019stochastic}

\begin{itemize}
  \item The membrane potential of a neuron takes finitely  many discrete
    values.
\item External synaptic input to each neuron is modeled by independent
  Poisson process. The rates of Poisson processes are identical in the
  same local population.
\item A neuron spikes when its membrane potential reaches a certain
  threshold. A post-spike neuron stays in a refractory state for an
  exponentially distributed random time.
\item The set of postsynaptic neurons is given by a graph $G$. 
\end{itemize}

The detailed description of this model is divided into the following
aspects. 

{\bf Neuron indices.} Consider a neuronal network model with $K$
local populations, denoted by
$L_{1}, \cdots, L_{K}$. Each local population has $N_{E}$ excitatory
neurons and $N_{I}$ inhibitory neurons. A type $Q$ neuron in
local population $L_{k}$ is denoted as neuron $(k, n, Q)$, where $k \in \{1,
\cdots, K\}$, $Q \in \{E, I\}$, and $n \in \{1, \cdots, N_{Q}\}$. The
triplet $(k, n, Q)$ is called the {\it label} of a neuron. We further
assign an integer-valued {\it index}, denoted by $\mathrm{id}_{(k, n,
  Q)}$, to neuron $(k, n, Q)$, such that
$$
  \mathrm{id}_{(k,n,Q)} = (k-1)(N_{E} + N_{I}) + n + N_{E} \mathbf{1}_{\{Q = E\}} \,.
$$
In other words, $\mathrm{id}$ is a function from the set of labels,
denoted by $\mathcal{L} = \{ (k, n, q) \,|\, 1 \leq k \leq K, 1 \leq n
\leq N_{E} + N_{I}, Q \in \{E, I\} \}$, to the set of indices, denoted by
$\mathcal{ID} = \{1, \cdots, K(N_{E} + N_{I})\}$. We call a neuron
with index $i$ ``neuron $i$'' when it does not lead to confusions. For each index $i \in \mathcal{ID}$, we
denote $\mathrm{Label}_{i}$ as the label of the corresponding neuron. Obviously
$\mathrm{Label}_{i}$ is the inverse function of $\mathrm{id}$ from the set of
indices to the set of labels. A connection graph $G$ is said to be
{\it static} if the edge set $E$ is fixed, and {\it random} if $E$ is
updated after each spike.

{\bf Connection graph.} The connection graph of the neural network is
a directed graph $G = (V, \mathcal{E})$ with $V = \mathcal{ID}$. The edge set
$\mathcal{E}$ is the collection of synapses such that
$$
  \mathcal{E} = \left \{  \{\mathrm{pre}, \mathrm{post}\} \,|\,  \mathrm{pre}
    \in V, \mathrm{post} \in V \right \} \,,
$$
where $\mathrm{pre}$ and $\mathrm{post}$ are indices of presynaptic
and postsynaptic neurons, respectively. When neuron $\mathrm{pre}$
fires a spike, its postsynaptic neuron $\mathrm{post}$ receives the
spike and changes its membrane potential. The collection of
postsynaptic neurons (resp. presynaptic neurons) of a neuron with index $i$ is denoted by
$\mathrm{Po}(i)$ (resp. $\mathrm{Pr}(i)$).

{\bf Neuron and external drive.}
We denote the membrane potential of a neuron with label $i \in \mathcal{ID}$ by
$$
V_{i} \in \Gamma := \{-M_{r}, -M_{r}+1, \dotsc , -1, 0, 1, 2, \dotsc, M\} \cup \{\mathcal{R}\},
$$
where $M, M_{r} \in \mathbb{N}_{+}$ denote the threshold for spiking
and the reversal potential, respectively. $\mathcal{R}$
represents the refractory state. When $V_{i}$ reaches $M$, a neuron fires a spike, and instantaneously moves to the refractory state
$\mathcal{R}$. At the refractory state, a neuron stays inactive for an
exponentially distributed amount of time with mean
$\tau_{\mathcal{R}}> 0$. After that, $V_{i}$ is reset to $0$.

The sources of stimulus that a neuron receives can be divided into
external drives and postsynaptic kicks from in-network neurons. The
external drive comes from outside of the neuronal network in this
model, either from sensory input or from other parts of the brain. We
assume that neurons in a local population receive external drive with
the same rate, and we model external drives to excitatory and inhibitory
neurons in local population $L_{l}$ by Poisson kicks with rates
$\lambda^{l}_{E}, \lambda^{l}_{I} > 0$ respectively, for $l \in \{4,6\}$. More precisely,
assume the label of neuron $i \in \mathcal{ID}$ is $\mathrm{Label}_{i}
= (k, n, Q)$. Then the time that
neuron $i$ receives external kicks is given by a Poisson process with rate
$\lambda^{k}_{Q}$. If $V_{i} \neq \mathcal{R}$ when neuron $i$
receives an external kick, $V_{i}$ immediately
increases (respectively decreased) by $1$.

{\bf Neuron spikes and postsynaptic kicks.}

When neuron $i$ fires a spike immediately after $V_{i}$ reaches the
threshold, all postsynaptic neurons in $\mathrm{Po}(i)$ receive a
postsynaptic kick. The effect of a postsynaptic kick is delayed by an
i.i.d. exponentially distributed random time with 
mean $\tau_{E}, \tau_{I} > 0$, for E, I kicks,
respectively. When the delay is over, the kick takes effect
instantaneously to neuron
$j \in \mathrm{Po}(i)$ if $V_{j} \neq \mathcal{R}$. After the
postsynaptic kick, $V_{j}$ jumps by $S_{QQ'}, Q, Q'
\in \{E, I\}$, where $Q$ and $Q'$
represent the neuron types of $j$ and $i$ respectively. $S_{QQ'}$ is
the strength of a postsynaptic kick, which is positive if $Q' = E$,
and negative if $Q' = I$. If $Q' = E$ and $V_{j}$ jumps to $\geq M$,
neuron $j$ jumps to $\mathcal{R}$ instead and fires a spike. If $Q' =
I$ and $V_{j}$ jumps to $< - M_{r}$, it takes value $V_{j} = -M_{r}$.

It remains to discuss non-integer $S_{Q,Q'}$. We let $S_{Q,Q'} = p +
u$ where $p = \lfloor{S_{Q,Q'}}\rfloor$ is the largest integer smaller
than $S_{Q,Q'}$ and $u$ is a Binomial random variable taking value in
$\{0,1\}$. 

{\bf Markov process.}

Because of the delay of postsynaptic kicks, the state of neuron $i$ is
denoted by a triplet $(V_{i}, H^{E}_{i}, H^{I}_{i})$. We use $H_i^E$
($H_i^I$ respectively) to store the number of E (I respectively) kicks
received from $\mathrm{Pr}(i)$ that have not taken effects. It is easy
to see that the model described above generates a Markov
process, denoted by $X_{t}$, on the state space 
$$
  \mathbf{\Omega} := (\Gamma \times \mathbb{N}_{+} \times \mathbb{N}_{+})^{K(N_{E} +
    N_{I})} \,.
$$
A state $\mathbf{x} \in \mathbf{\Omega}$ has the form
$$
  \mathbf{x} = \{(V_{i}, H^{E}_{i}, H^{I}_{i})\}_{i \in \mathcal{ID}} \,.
$$
The transition probability of $X_{t}$ is
$$
  P^{t}( \mathbf{x}, \mathbf{y})  = \mathbb{P}[ X_{t} = \mathbf{y}
  \,|\, X_{0} = \mathbf{x}] \,.
$$
A probability measure $\pi$ on $\mathbf{\Omega}$  is said to be
invariant if $\pi = \pi P^{t}$, where $\pi P^{t}$ is given by the left
operator of $P^{t}$
$$
  \pi P^{t}( \mathbf{x}) = \sum_{ \mathbf{y} \in \mathbf{\Omega}} \pi(
  \mathbf{y}) P^{t}( \mathbf{y}, \mathbf{x} ) \,.
$$
Throughout this paper, we denote the conditional probability
$\mathbb{P}[\cdot \,|\, \mathrm{law}(X_{0}) = \mu]$ by
$P_{\mu}[\cdot]$ for the sake of simplicity, where $\mu$ is a
probability measure on $\mathbf{\Omega}$, and $\mathrm{law}(X_{0}) =
\mu$ means the initial distribution of $X_{t}$ is $\mu$. 

\subsection{Visual cortex models I-III}
We use the following three visual cortex models, from simple to
complicated, as the numerical examples of this paper. The first model
({\bf Model I}) only has two local populations, one feedforward layer
and one feedback layer. No
spatial factor is considered in this model. The
second model ({\bf Model II}) has one feedforward layer and one
feedback layer. Each layer consists of $16$ local populations
(hypercolumns). Each neuron in {\bf Model II} has a coordinate. And
the connection graph $G$ is generated according to locations of
neurons. {\bf Model III}, the most complex model, has the same layers
, hypercolumns, and neuron coordinates as in {\bf Model II}. Each hypercolumn of {\bf Model
  III} has $4$ orientational columns. The connection graph $G$ depends
on both location and orientation. We use easier model to illustrate
entropy and mutual information, and discuss the role of mutual
information by showing numerical results for more complicated
models. {\bf Model III} is used to demonstrate degeneracy and
complexity, which are two systematic measures defined on complex
biological networks. In all three models, we choose $M = 100$, $M_{r}
= 66$, $N_{E} = 300$, $N_{I} = 100$, $\tau_{\mathcal{R}} = 2.5$ ms,
$S_{EE} = 5.0$, $S_{IE} = 2.3$, $S_{EI} = -3.5$, and $S_{II} = -3.0$
unless further specified. Other
parameters including the number of local populations $K$, the
connection graph $G$, delay times $\tau_{E}$ and $\tau_{I}$, and
external drive rates $\lambda^{1}_{E}, \lambda^{1}_{I}, \cdots,
\lambda^{K}_{E}, \lambda^{K}_{I}$, will be prescribed when introducing
each model.

{\bf Model I.} In the first model we have $K = 2$. Two local
populations represent the feedforward and feedback layer
respectively. We set $\lambda^{i}_{Q} = 5000$ for $i = 1,2$ and $Q =
E, I$. The connection graph $G$ is a spatially homogeneous
random graph. For each pair of $i, j \in \mathcal{ID}$, let
$\mathrm{Label}_{i} = (k_{i}, n_{i}, Q_{i})$ and $\mathrm{Label}_{j} =
(k_{j}, n_{j}, Q_{j})$. The connection probability is divided into
three cases: (i) $\{i, j\} \in \mathcal{E}$ with probability $P_{Q_{i}Q_{j}}$ if
$k_{i} = k_{j}$ (intra-layer connection), (ii) $\{i, j\} \in \mathcal{E}$ with probability $P^{f}_{Q_{i}Q_{j}}$ if
$k_{i} = 1,  k_{j} = 2$ (feedforward connection), and (iii) $\{i, j\}
\in \mathcal{E}$ with probability $P^{b}_{Q_{i}Q_{j}}$ if $k_{i} = 2,  k_{j} =
1$ (feedback connection). In this model, we choose parameters $P_{EE}
= 0.15$,  $P_{IE} = 0.5$, $P_{EI} = 0.5$, and $P_{II} = 0.4$, $P^{f}_{EI}
= P^{f}_{II} = P^{b}_{EI} = P^{b}_{II} = 0$. In other words inhibitory
neurons only connect to neurons in the same layer. We set $P^{f}_{EE}
= \rho_{f} P_{EE}$, $P^{f}_{IE} = \rho_{f} P_{IE}$, $P^{b}_{EE}
= \rho_{b} P_{EE}$, and $P^{b}_{IE} = \rho_{b} P_{IE}$. $\rho_{f}$ and
$\rho_{b}$ represent the strength of feedforward and feedback
connection respectively. Without further specification, synapse delay
times are chosen to be $\tau_{E} = 2.0$ ms and $\tau_{I} = 4.5$ ms. $\rho_{b}$ and $\rho_{f}$ are two main control parameters in {\bf
Model I}.

{\bf Model II.} In the second model we have $K = 32$, with the first $16$ local
populations as hypercolums at the feedforward layer and the rest are
hypercolumns in the feedback layer. We set $\lambda^{i}_{Q} = 5000$
for $i = 1,2, \cdots, 15, 16$ and $\lambda^{i}_{Q} = 4500$ for $i =
17,18, \cdots, 31, 32$, $Q = E, I$. Each neuron has a location
coordinate. We assume neurons in a local population form a $10\times
10$ lattice. At each lattice point, there are three E neurons and one
I neuron. We assume one local population represents neurons in one
hypercolumn of the visual cortex with a size $0.25$ mm. In other words
the grid size of this lattice is $25$ $\mu$m, and the boundary neurons
are $12.5$ $\mu$m away from their nearest local boundary. Local
populations in each layer is a $4\times 4$ array. The connection graph
is based on the types and locations of each pair of neurons. See
Figure \ref{fig1} for the layout of this model. 

For each pair of $i, j \in \mathcal{ID}$, let
$\mathrm{Label}_{i} = (k_{i}, n_{i}, Q_{i})$ and $\mathrm{Label}_{j} =
(k_{j}, n_{j}, Q_{j})$. The connection radius belongs to one of the
three cases: (i) $\{i, j\} \in \mathcal{E}$ with radius $L_{Q_{i}Q_{j}}$ if
$k_{i} $ and $k_{j}$ are either both $\leq 16$ or both $> 16$ (intra-layer connection), (ii) $\{i, j\} \in \mathcal{E}$ with radius $L^{f}_{Q_{i}Q_{j}}$ if
$k_{i} \leq 16,  k_{j} > 16$ (feedforward connection), and (iii) $\{i, j\}
\in \mathcal{E}$ with radius $L^{b}_{Q_{i}Q_{j}}$ if $k_{i} > 16,
k_{j} \leq 16$ (feedback connection). We choose $L_{EE}
= L_{IE} = 0.15$, $L_{EI} = L_{II} = 0.10$, $L^{f}_{EI}
= L^{f}_{II} = L^{b}_{EI} = L^{b}_{II} = 0.10$. 

The actual connection probability is the product of a baseline connection probability
$p_{Q_{i}Q_{j}}$ (resp. $p^{f}_{Q_{i}Q_{j}}$, $p^{b}_{Q_{i}Q_{j}}$) as defined in {\bf Model I} and the probability density
function of a normal distribution with a standard deviation
$L_{Q_{i}Q_{j}}$ (resp. $L^{f}_{Q_{i}Q_{j}}$, $L^{b}_{Q_{i}Q_{j}}$) . More precisely, the probability that neurons $i$ and
$j$ at the same layer are connected is  
$$
p_{Q_{i}Q_{j}}\frac{\exp(-\frac{d^2}{2L_{Q_{i}Q_{j}}^2})}{(2\pi L_{Q_{i}Q_{j}}^2)}
$$
where $d$ is the distance between the two neurons in question. Cases of feedforward and feedback connection probability are
analogous. Baseline connection probabilities are $p_{EE} = 0.01$,
$p_{IE} = 0.035$, $p_{EI} = 0.03$, $p_{II} = 0.03$. Again, I neurons
only have intra-layer connections so $p^{f}_{EI}
= p^{f}_{II} = p^{b}_{EI} = p^{b}_{II} = 0$. The rule of
feedforward and feedback connection probability is analogous to those of Model
I. We have $p^{f}_{QE} = \rho^{f}p_{QE}$ and $p^{b}_{QE} =
\rho^{b}p_{QE}$ for $Q = I, E$. $\rho^{f}$ and $\rho^{b}$ are two
control parameters.

\begin{figure}[!h]
\centering
\includegraphics[width=0.49\linewidth]{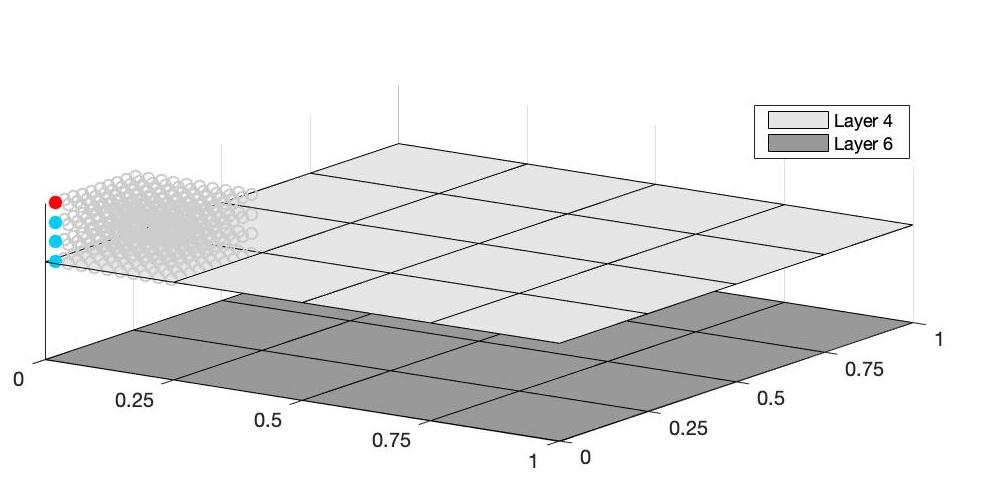}
\includegraphics[width=0.49\linewidth]{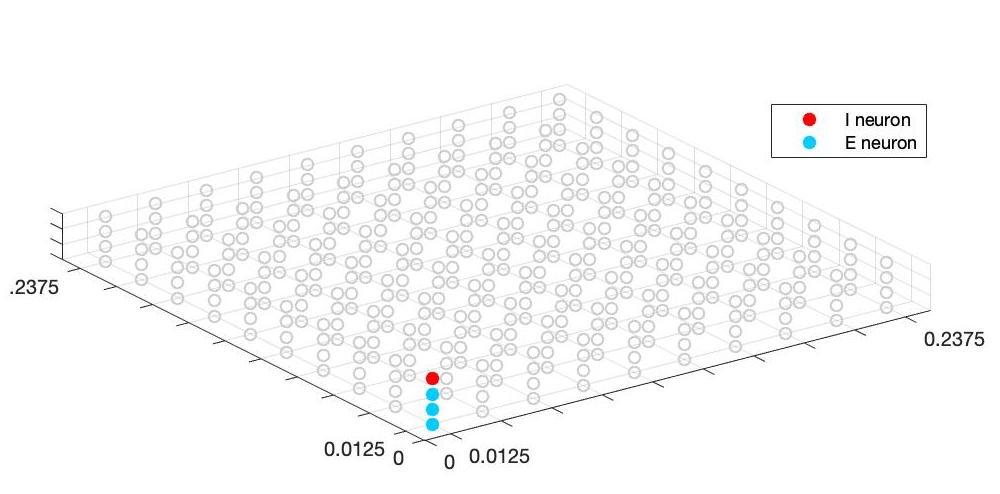}
\caption{Left: Layout of Model II. Each layer consists of $4\times 4$
  hypercolumns. Right: Geometry of each hypercolumn. Each lattice
  point is occupied by 3 excitatory neurons and 1 inhibitory neuron. }
\label{fig1}
\end{figure}

{\bf Model III.} Then we add orientation columns into the model in
order to model two layers in the primary visual cortex. More
precisely, each hypercolumn in Model II is further divided into four
orientational columns that resembles the pinwheel structure
\cite{hubel1995eye, kang2003mexican, kaschube2010universality}. The layout of orientational columns is
demonstrated in Figure \ref{fig2}. We assume the visual stimulation is
vertical. The external drive rates to orientational columns with
preferences $0$ deg, $45$ deg, $90$ deg, and $135$ deg are multiplied by
coefficient $1.0$, $0.6$, $0.2$, and $0.6$ respectively. 

We also add long-range excitatory connections on the feedback
layer that hits neurons with same orientation preference, with is
consistent with experimental studies \cite{stettler2002lateral,
  gilbert1989columnar, malach1993relationship} that many 
connections are between neurons with the same orientation preferences. The
connection probability is given by a linear function. For a pair of
neurons $i, j \in \mathcal{ID}$, the probability of having long range
connection is 
$$
p_{Q_{i}Q_{j}}C_{long}(1 - 0.5 d),
$$
if $Q_{i} = E$, $d > 2 L_{EE}$, and neuron $i, j$ have the same
orientation preference, where the baseline connection probability
$p_{Q_{i}Q_{j}}$ is the same as before, $d$ is the distance between the two
neurons in question, and $C_{long}$ is a changing parameter that
controls the strength of long range connections. Here we assume that
the probability of having long range connections decreases linearly
and becomes zero if $d$ is larger than $2$ mm.

Other configurations of Model III are identical to those of Model II. 

\vskip 1cm
\begin{figure}[!h]
    \centering
    \includegraphics[width=0.8\linewidth]{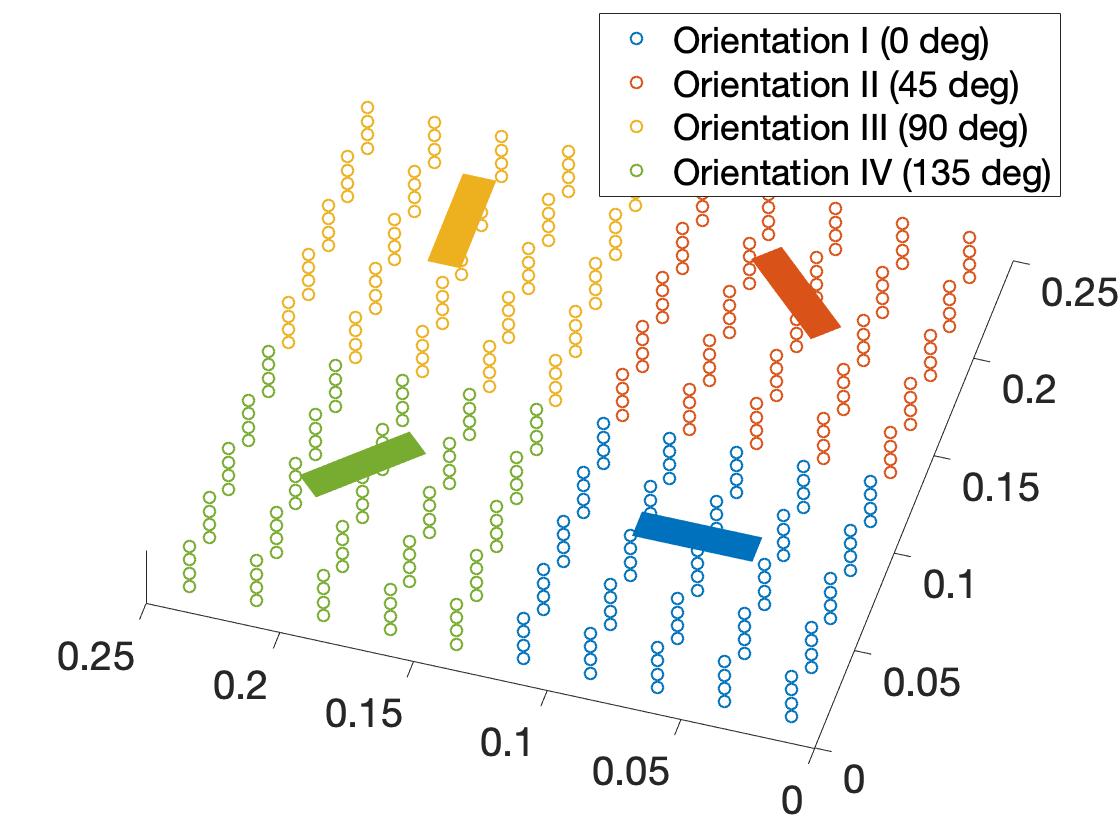}
    \caption{The layout of neurons in one hypercolumn.}
    \label{fig2}
\end{figure}

\vskip 1cm

\subsection{Firing rate and spiking pattern in models}

In this subsection we will show some simulation results about our
visual cortex models. The first result is about the mean firing rate. Figure \ref{fig3} gives the firing rate of {\bf
  Model I} versus the background drive rate $\lambda^{i}_{E} =
\lambda^{i}_{I} = \lambda$ for $i = 1, 2$ when $\lambda$ increases
from $1000$ to $8000$. We can see an increase of empirical firing rate
with the background drive. This is consistent with our previous
results in \cite{li2019well, li2019stochastic}. The heat map of firing rates in the most
complicated model (Model III) is demonstrated in Figure \ref{fig4}. We
can clearly distinguish orientation columns in the heat map. The
vertical-preferred orientation columns have highest firing rate, while
the horizontal-preferred orientation columns fire the slowest. See caption of
Figure \ref{fig4} for the choice of control parameters.

\begin{figure}[!h]
    \centering
    \includegraphics[scale=0.25]{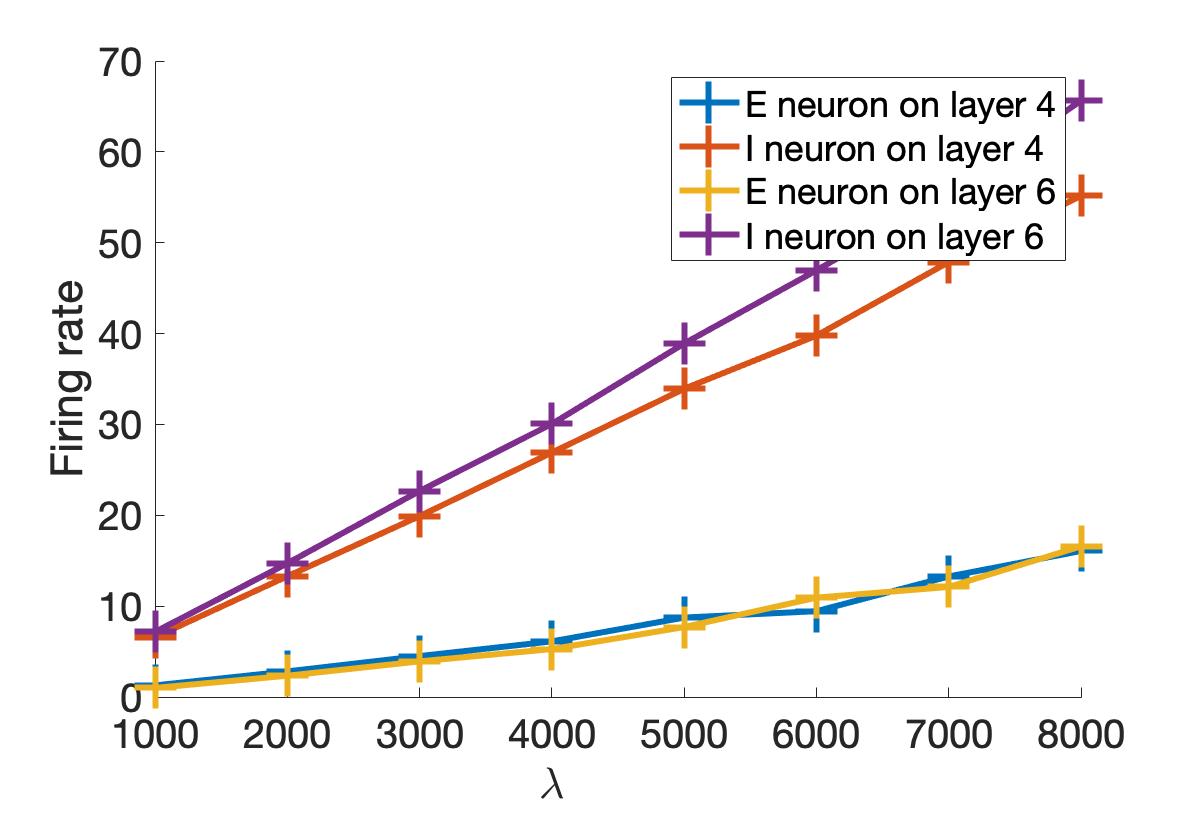}
    \caption{E-population and I-population mean firing rates in layer 4 and 6 when $\lambda$
changes from $1000$ to $8000$. Parameters are $\rho_{f} = \rho_{b} = 0.6$,
$\tau_{E} = 2.0$, $\tau_{I} = 4.5$. }
    \label{fig3}
\end{figure}

\begin{figure}[!h]
    \centering
    \includegraphics[width=\linewidth]{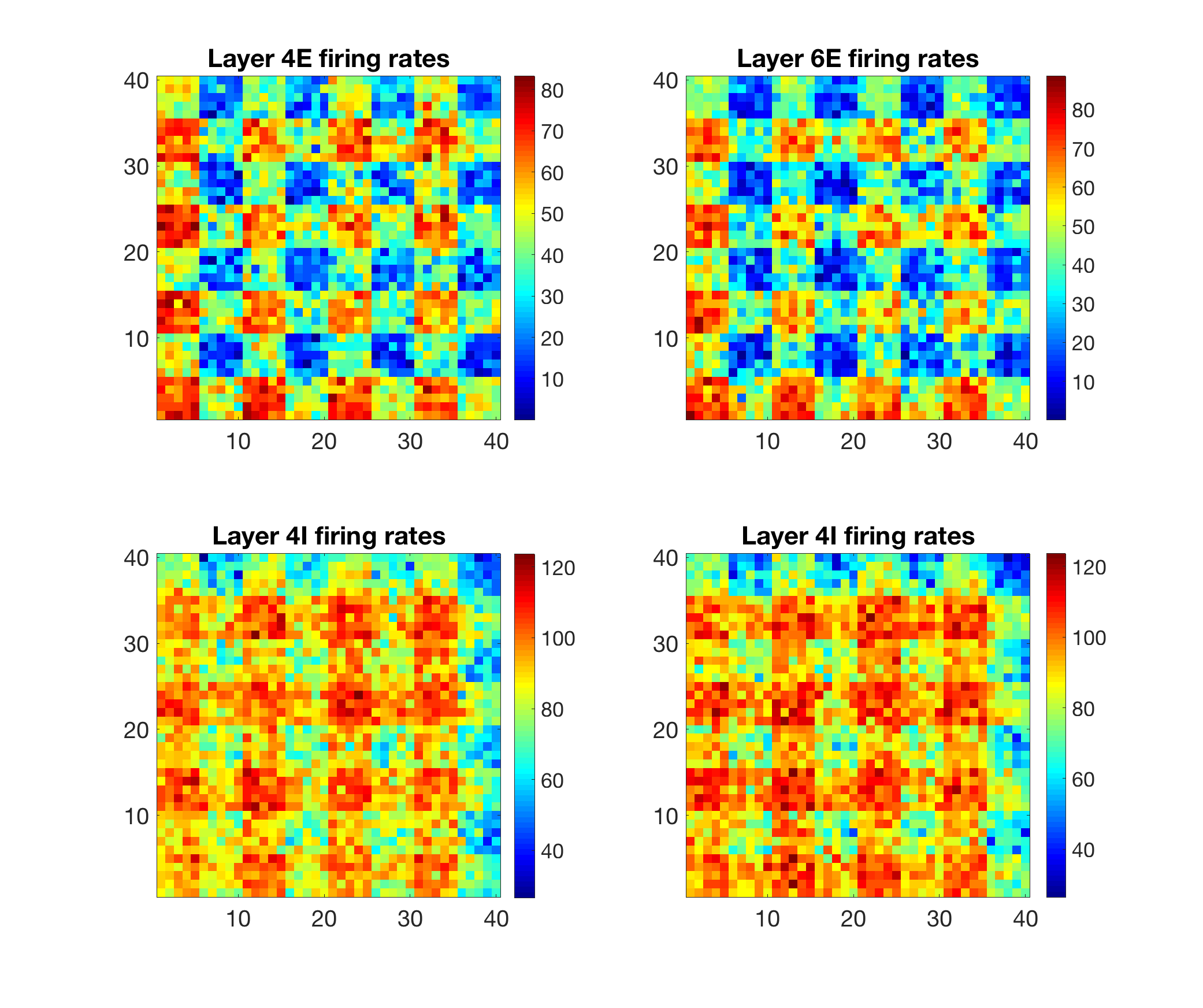}
    \caption{Heat map for E and I firing rates of Model
      III. Parameters are $\rho_{f} =
\rho_{b} = 0.6$, $\tau_{E} = 2.0$, $\tau_{I} = 4.5$. }
    \label{fig4}
\end{figure}

\vskip 1cm

The spiking pattern is another important feature of spiking neuron
models, as neurons pass information through spike trains. Similar as
in many pervious studies, the visual cortex models in our study
exhibit partial synchronizations. Due to recurrent excitation, neurons
tend to form a series of spike volleys, each of which consists of
spikes from a proportion of the total population. This phenomenon is
known as the multiple firing event (MFE), which is believed to be
responsible for the Gamma rhythm \cite{rangan2013dynamics, rangan2013emergent, chariker2015emergent}. Figure \ref{fig5} demonstrates the
raster plot produced by Model I. We can see obvious MFEs that lie between
homogeneous spiking activities and full synchronizations. In addition,
the main control parameter of MFE sizes is the ratio
$\tau_{I}/\tau_{E}$. Higher $\tau_{I}$-to-$\tau_{E}$ ratio is
responsible for more synchronized spike activities. The mechanism of the dependence of MFE sizes on
the synapse delay time is addressed in \cite{li2019stochastic}. Slower $\tau_{I}$
means E-cascade can last longer time, which contributes to larger MFE
sizes. In Figure \ref{fig6}, we can see raster plots of $(2, 2)$
hypercolumn in the feedforward and feedback layer produced by Model
II. Top and bottom plots are the cases with and without synapse connections
between these two layers, respectively. When the feedforward and feedback connection is turned on,
one can see some correlations between MFEs. Later we will use mutual
information to quantify such correlation. 

\vskip 1cm
\begin{figure}[!h]
    \centering
    \includegraphics[scale=0.18]{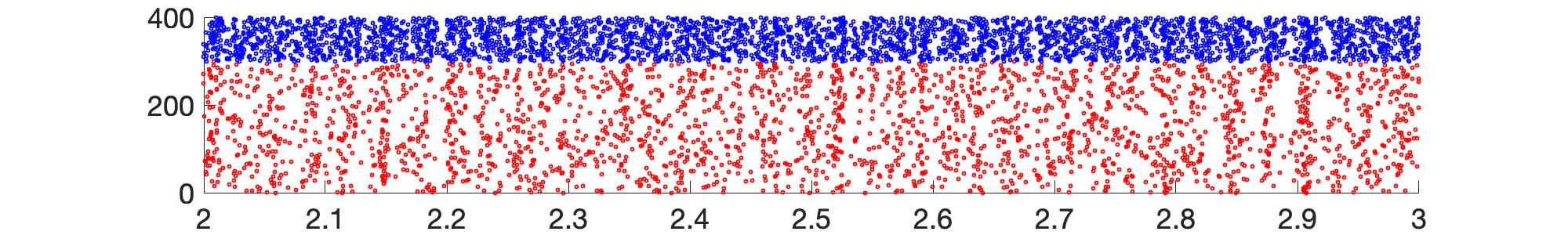}
    \includegraphics[scale=0.18]{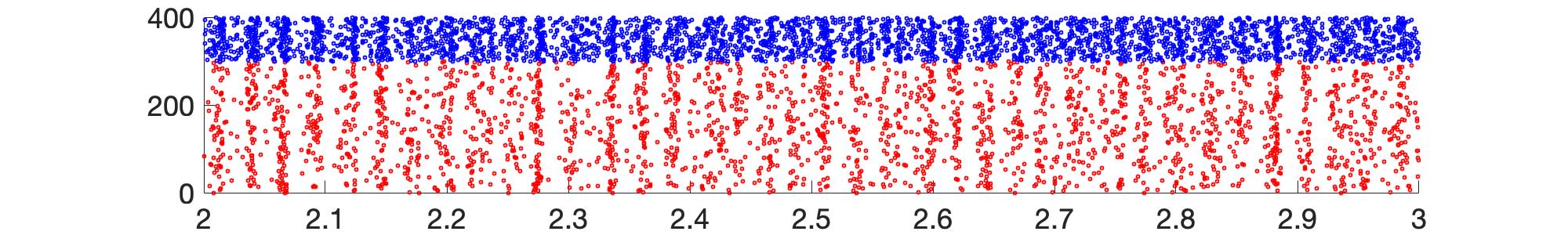}
    \includegraphics[scale=0.18]{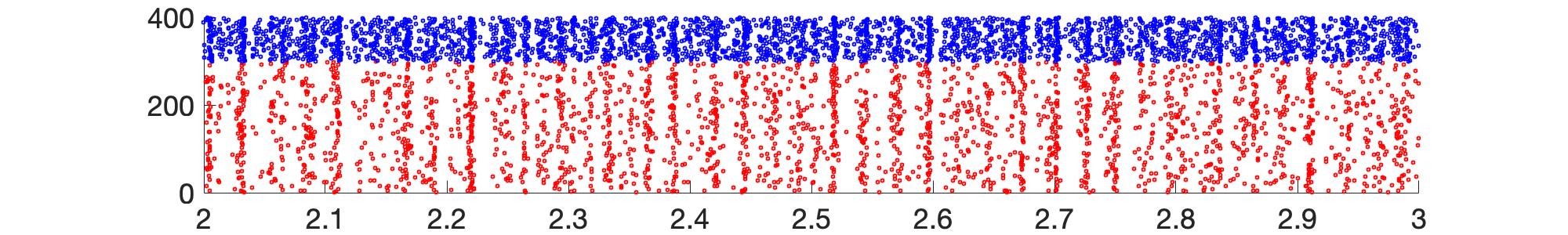}
    \caption{Three raster plots for layer 4 of Model I without feedforward or feedback
connection. Top: $\tau_{E} = 4.5$ ms, $\tau_{I} = 4.5$ ms. Middle: $\tau_{E} =
2.0$ ms, $\tau_{I} = 4.5$ ms. Bottom: $\tau_{E} = 0.9$ ms, $\tau_{I} = 4.5$ ms.}
    \label{fig5}
\end{figure}

\begin{figure}[!h]
    \includegraphics[width = \linewidth]{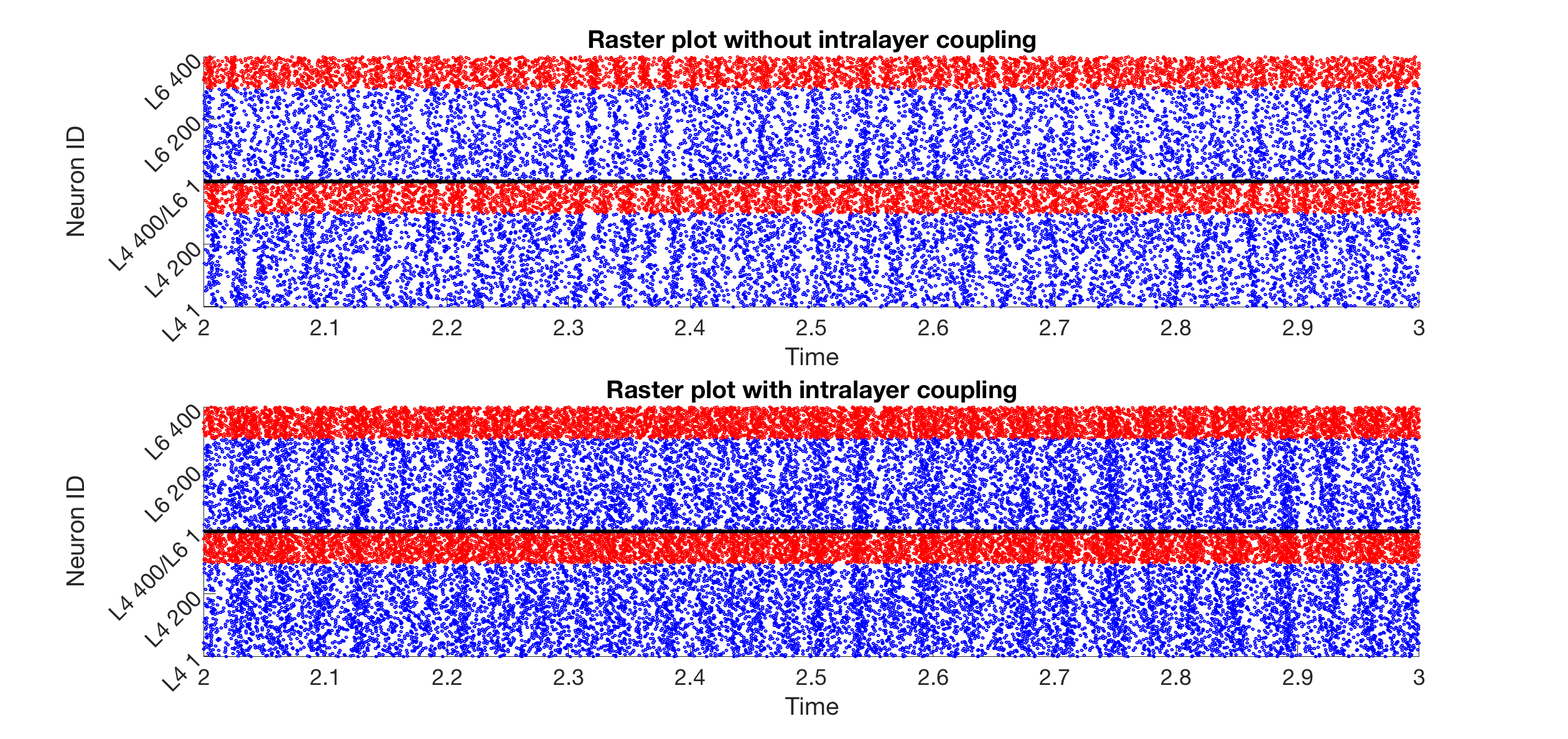}
    \caption{Top. Two raster plots for the (2,2) hypercolumn in layer 4 (top) and
the same hypercolumn in layer 6 (bottom). $\rho_{b} = \rho_{f} = 0$. $\tau_{E} = 2.0$
ms, $\tau_{I} = 4.5$ ms. Time span 2-3s. b. Same parameters but
with intra-layer connections $\rho_{b} = \rho_{f} = 0.6$. } 
    \label{fig6}
\end{figure}

\vskip 1cm

\section{Coarse-grained entropy and law of large numbers}

\subsection{Coarse-grained entropy and coarse-grained information theoretic measures.}
Let $T > 0$ be a time window size that is associated to a coarse-grained
information theoretic measure. Let
$$
  \mathcal{SP} = \{ (t_{i}, \xi_{i}) \,|\, t_{i} \in \mathbb{R}_{+},
  \xi_{i} \in \mathcal{ID} \}
$$
be the configuration space of neuron spikes, where $i$ is the index of a spike, $t_{i}$ is the time
of a spike, and $\xi_{i}$ is the index of the spiking neuron.  Let
$S_{[0, T])} = (s_{1}, \cdots, s_{Z})$ be a
spike train produced by $X_{t}$ within the time window $[0, T)$, where $s_{i} \in
\mathcal{SP}$ denotes the $i$-th spike within this time window. $Z$
is a random variable that represents the total spike count on $[0,
T)$. Let
$$
  \mathcal{SPT} = E \cup \bigcup_{n \geq 1} \mathcal{SP}^{n} 
$$
be the configuration space of neural spike trains, where $E$
represents an empty spike train, and $\mathcal{SP}^{n}$ is the $n$-product
set of $\mathcal{SP}$. Finally, we let $\mathbf{D} = \{1, 2, \cdots, \Sigma
\}$ be a dictionary set
that consists of countably many distinct integers $1, \cdots, \Sigma$. ($\Sigma$
could be infinity.) A function
$\zeta: \mathcal{SPT} \rightarrow \mathbf{D}$ is a coarse-grained
mapping that maps a spike train to an element in $\mathbf{D}$. For
example, the simplest coarse-grained function is the number of spikes
of $X_{t}$ on the interval $[0, T)$. In this case we have $\Sigma =
\infty$.

For $i \in \mathbf{D}$, let 
$$
  p_{i} = \mathbb{P}_{\pi}[\zeta(S_{[0, T)}) = i]
$$
be the probability that the coarse-grained function maps a spike train
$S_{[0, T)}$ to $i$ when starting from the invariant probability
measure $\pi$. The {\it coarse-grained entropy} with respect to $T$
and $\zeta$ is 
$$
  H_{T, \zeta} = -\sum_{i \in \mathbf{D}} p_{i}\log p_{i} \,.
$$
If there are $q$ coarse-grained functions $\zeta_{1}, \cdots,
\zeta_{q}$ and a function $F: \mathbb{R}^{q} \rightarrow \mathbb{R}$,
an {\it information theoretical measure} $\mathfrak{M}$ with respect to $T$,
$\zeta_{1}, \cdots, \zeta_{q}$, and $F$ is given by 
$$
  \mathfrak{M} = F(H_{T, \zeta_{1}}, \cdots, H_{T, \zeta_{q}}) \,.
$$

\begin{rem}
The reason of defining the coarse-grained entropy is to make entropy
and information theoretical measures computable. The classical
definition of neural entropy can only allow at most one spike in each
bin (time window). In a large neural network, neuronal activities are
usually synchronized in some degree. As a result, the necessary time
window size quickly becomes too small to be practical. 

The main simplification we made is to ignore the order of spikes that
are sufficiently close to each other. Needless to say this treatment
loses some information. But it makes information theoretical measures
more computable for large networks. In addition, we map a spike train
within a time window to an integer to further reduce the state
space. This is because the state space of naive spike counting is
still huge. If we naively count the spikes in each hypercolumn in a
model with $K$ local populations, there will be $(N_{E} + N_{I})^{K}$
possible spike counting results even if we assume a neuron cannot
spike twice in a time bin.  
\end{rem}

The definition of the coarse-grained entropy is very general. In
practice, the function $\zeta$ can be given in the following ways to
address different features of the neuronal network. 

\begin{itemize}
  \item[A.] Spike counting for a certain local population. If we are
    interested in the information produced by a certain local
    population, say local population $k$, we can have
$$
  \zeta( S_{[0, T)}) = \zeta \left(\{(t_{1}, \xi_{1})\cdots, (t_{Z},
    \xi_{Z}) \}\right) = \sum_{j = 1}^{Z}
  \mathbf{1}_{\{\mathrm{Label}_{\xi_{j}}(1) = k \}} \,,
$$
where $\mathrm{Label}_{\xi_{j}}(1)$ means the first entry of
$\mathrm{Label}_{\xi_{j}}$, which is the index of the local population of
spike $j$. Here one can replace $k$ by either a set of local
populations, or restrict on a certain type of neurons.
\item[B.] Coarse-grained spike counting. If the state space of spike
  counting is too big to estimate accurately, we can introduce a
  partition function $\Theta: \mathbb{Z}_{\geq 0} \rightarrow \{1, \cdots,
  d\}$, such that $\Theta(n) = i$ if $a_{i} \leq n < a_{i+1}$, where $0 =
  a_{0} < a_{1} < \cdots < a_{d+1} = \infty$ is a given sequence of
  numbers called a ``dictionary''. Then let
 $$
  \zeta( S_{[0, T)}) = \Theta \left (\sum_{j = 1}^{Z}
  \mathbf{1}_{\{\mathrm{Label}_{\xi_{j}}(1) = k\}} \right )\,.
$$
\item[C. ] Spike counting with delays. If one would like to address
  time lags of spike activities, the time window $[0, T)$ can be
  further evenly divided into many sub-windows $[0, T/m), [T/m, 2T/m),
  \cdots, [(m-1)T/m, T)$. Assume we still adopt the coarse grained
  spike counting in {\it B}, we have
$$
  \zeta( S_{[0, T)}) = \sum_{l = 0}^{m-1}d^{l}\Theta\left (\sum_{j = 1}^{Z}
  \mathbf{1}_{\{\mathrm{Label}_{\xi_{j}}(1) = k, t_{j} \in [lT/m, (l+1)T/m)\}} \right )\,.
$$
In other words we take consideration of spike counts in each
sub-windows. Integer $m$ is said to be the ``word length''.
\end{itemize}

\subsection{Stochastic stability and law of large numbers.}
The aim of this section is to prove the law of large numbers holds for
entropy. In other words the entropy is computable by running Monte
Carlo simulations. To do so, we first need to show the stochastic
stability of $X_{t}$.

Let 
$$
  U(\mathbf{x}) = 1 + \sum_{i \in \mathcal{ID}} (H_{i}^{E} + H_{i}^{I})
$$
be a function on $\mathbf{\Omega}$. For any signed distribution $\mu$
on the Borel $\sigma$-algebra of $\Omega$, let
$$
  \|\mu\|_{U} = \sum_{\mathbf{x} \in \mathbf{\Omega}}U(\mathbf{x})|
  \mu( \mathbf{x})| 
$$
be the $U$-weighted total variation norm, and let 
$$
  L_{U}( \Omega) = \{ \mu \mbox{ on } \mathbf{\Omega} \,|\, \|
  \mu\|_{U} < \infty \}
$$ 
be
the collection of probability distributions with finite $U$-norm. In
addition, for any function $\eta( \mathbf{x})$ on $\Omega$, denote
$$
  \sup_{\mathbf{x} \in \Omega} \frac{|\eta ( \mathbf{x})|}{U( \mathbf{x})}
$$
by the $U$-weighted supremum norm.

\begin{thm}
\label{thm1}
$X_{t}$ admits a unique invariant probability distribution $\pi \in
L_{U}( \mathbf{X})$. In addition, there exists positive constants
$c_{1}, c_{2}$ and $\gamma \in (0, 1)$ such that
$$
  \| \mu_{1} P^{t} - \mu_{2} P^{t} \|_{U} \leq c_{1} \gamma^{t}
  \|\mu_{1} - \mu_{2} \|_{U}
$$
for any $\mu_{1}, \mu_{2} \in L_{U}( \Omega)$ and
$$
  \|P^{t} \eta - \pi(\eta) \|_{U} \leq c_{2}\gamma^{t}\| \eta - \pi(\eta)\|_{U}
$$
for any test function $\eta$ with $\|\eta\|_{U} < \infty$. 
\end{thm}

We skip the proof of Theorem \ref{thm1} here because it is almost
identical to that of Theorem 3.1 in \cite{li2019stochastic}. The only difference
is that postsynaptic neurons in the model of \cite{li2019stochastic} is chosen
randomly {\it after} a spike, which does not affect the proof. To prove the existence of an invariant probability
measure $\pi$ and the exponential convergence to it, we need to
construct a Lyapunov function $U( \mathbf{x})$ such that
$$
  P^{h}U \leq \gamma U + K
$$
for some $h > 0$, $\gamma \in (0, 1)$ and $K < \infty$. In addition the
``bottom'' of $U$, denoted by $C = \{ \mathbf{x} \in \Omega \,|\, U(
\mathbf{x}) \leq R \}$ for some $R > 2K(1 - \gamma)^{-1}$, must
satisfies the minorization condition, which means there exists a
constant $\alpha > 0$ and a probability distribution on $\Omega$, such
that $P^{h}(\mathbf{x}, \cdot) \geq \alpha \nu(\cdot)$ uniformly for
$\mathbf{x} \in C$. Then the existence and uniqueness of $\pi$ and the
exponential speed of convergence to $\pi$ follows from
\cite{li2019stochastic}. We refer the proof of Theorem 3.1 in \cite{li2019stochastic} for
the full detail.

The next Theorem is about the computability of information-theoretical
measures. The way of estimating $H_{T, \zeta} $ is called ``plug-in
estimate'', which is known to convergence if the spike train $S_{[0,
  T)}$ is i.i.d. sampled from a certain given probability distribution
\cite{antos2001convergence} (also see \cite{paninski2003estimation}). However, in our case $S_{[0, T)}$ is the spike count
generated by a Markov chain, which depends on the sample path of
$X_{t}$ on $[0, T]$. Therefore, one needs to sample $S_{[0, T)}$ by
running $X_{t}$ over many time intervals $[0, T), [T, 2T),
\cdots$. Hence the result in \cite{antos2001convergence} does not apply
directly. Instead, we need to use the concept of {\it sample path
  dependent observable} to show that $H_{T,
  \zeta}$ is a computable quantity.

\begin{thm}
\label{thm2}
Let $\mathcal{T} = NT$ be the length of a trajectory of $X_{t}$. For any coarse-grained mapping $\zeta$ and any $i \in \mathbf{D}$,
denote the empirical probability of $\zeta = i$ by
$$
  \hat{p}_{i} = \frac{1}{N} \sum_{j = 1}^{N} \mathbf{1}_{\{
    \zeta(S_{[(j-1)T, jT)}) = i\}}(X_{t}) \,.
$$
We have 
$$
  \lim_{N\rightarrow \infty} \hat{p}_{i} =
  \mathbb{P}_{\pi}[\zeta(S_{[0, T)}) = i] = p_{i} \,. 
$$
\end{thm}
\begin{proof}
The proof of this theorem relies on the concept of {\it Markov sample
  path dependent observables}. Let $X_{t}$ be a Markov process. A
function $Y$ is said to be a Markov sample-path dependent observable on
an interval $[t_{1}, t_{2})$ if 
\begin{enumerate}
  \item[i] $Y$ is a real-valued function on $C_{\Omega}([t_{1},
    t_{2}))$, where $C_{\Omega}([t_{1},
    t_{2}))$ is the collection of cadlag paths from $t_{1}$ to $t_{2}$
    on $\Omega$.
\item[ii] The law of $Y$ only depends on the value of $X_{t_{1}}$.
\end{enumerate}

Now let $Y_{1}, Y_{2}, \cdots, Y_{n}, \cdots$ be a sequence of
sample-path dependent observables on $[0, T)$, $[T, 2T)$, $\cdots$,
$[(n-1)T, nT)$ respectively. Since the law of $X_{t}$ converges to
$\pi$ as $t \rightarrow \infty$, by Theorem 3.3 in \cite{li2019stochastic}, we
have the law of large numbers for $\{Y_{n}\}$, i.e., 
$$
  \lim_{N \rightarrow \infty} \frac{1}{N}\sum_{n = 1}^{N}Y_{n} =
  \mathbb{E}_{\pi}[Y_{1}] 
$$
provided there exists a constant $M < \infty$, such that
$\mathbb{E}[Y^{2}_{n} \,|\, X_{(n-1)T} = \mathbf{x}] < M$ uniformly
for all $\mathbf{x} \in \mathbf{\Omega}$.

Hence we only need to construct a sequence of sample-path dependent
observables $\{Y_{n}\}$ whose expectations equal to $\hat{p}_{i}$. Let 
$$
  Y_{n} = \mathbf{1}_{\{\zeta(S_{[(n-1)T, nT)}) = i\}} \,.
$$
It is easy to see that $Y_{n}$ is Markov because $X_{t}$ is a Markov
process. In addition, $\mathbb{E}[Y_{n}^{2}]\leq 1$ uniformly because
$Y_{n}$ is a indicator function. Therefore, $\{Y_{n}\}_{n =
  1}^{\infty}$ satisfies the law of large numbers. This completes the
proof. 

\end{proof}

\subsection{Numerical result and discussion}
We did the following two numerical simulations to demonstrate the role
of coarse-grained entropy in our network models. Consider model I
without feedforward or feedback connections. Without loss of
generality we only compute the entropy of the feedforward layer. The
synapse delay times are chosen to be $\tau_{E} = 2$ ms and $\tau_{I} =
4.5$ ms. 

The first test is about the entropy rate per second with
increasing word length. We use spike counting with delays to define
the coarse-grained mapping $\zeta$ (case C in Section 3.1). The
duration of each sub-window is $T/m = 5$ ms. The time window size goes
through $T = 5 $ ms (m=1) to $T = 60$ ms (m = 12). Then we define the partition
function $\Theta: \mathbf{Z}_{+} \rightarrow \{0, 1, 2, 3\}$ with
$\Theta(n) = 0$ for $n < 4$, $\Theta(n) = 1$ for $4 \leq n < 8$,
$\Theta(n) = 2$ for $8 \leq n < 12$, and $\Theta(n) = 3$ for $n \geq
12$. Figure \ref{fig7} shows the entropy rate (entropy divided by the
time window size) versus time window size. The curve is similar to
that given in \cite{strong1998entropy}. Starting from word length $=2$, and until
the word length being too long for estimating entropy, the entropy rate declines
linearly. But when a ``word'' is too long such that the entropy estimation
is badly under-sampled, the estimated entropy rate declines
faster-than-linear, as explained in \cite{strong1998entropy}. We use
extrapolation to estimate the entropy rate at the infinite window size
limit, as shown in Figure 3 of \cite{strong1998entropy}.  

\begin{figure}[!h]
\centerline{\includegraphics[width = 0.6\linewidth]{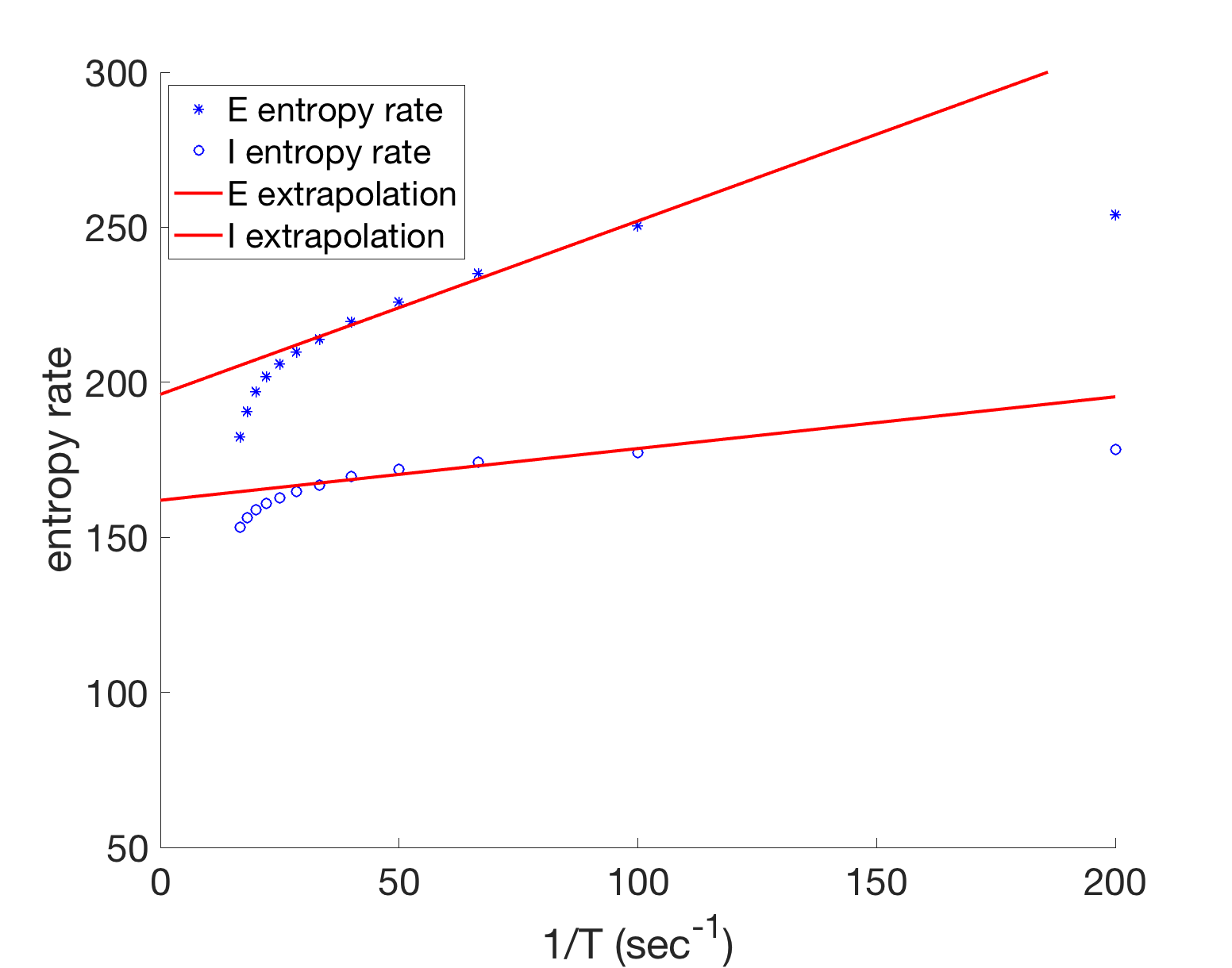}}
\caption{Entropy rate versus time window length. Samples are collected
  from eight trajectories with length $10000$. Linear extrapolation
  uses points with respect to $m = 2, 3, \cdots, 7$. }
\label{fig7}
\end{figure}

The second test is the entropy for different degrees of
synchrony. Here we only consider word length $ = 1$ with time window size
$T = 15$ ms, with a more refined partition function $\Theta:
\mathbb{Z}_{+} \rightarrow \{0, 1, \cdots, 45\}$ such that $\Theta(n)
= i$ for $3i \leq n < 3(i+1)$ if $i = 0, \cdots, 44$ and $\Theta(n) =
45$ if $n \geq 135$. The excitatory synapse delay time is set to be $\tau_{I} =
4.5$ ms. And $\tau_{I}$-to-$\tau_{E}$ ratio varies between $0.5$ and
$3.5$. In addition we let $S_{EE} = 6$ to increase the degree of synchrony. The entropy rate versus $\tau_{I}/\tau_{E}$ is plotted in
Figure \ref{fig8}. We can see that the entropy reaches a peak in the middle
of the tuning curve for both E and I neurons. The peak locations are
different for E and I neurons partially due to different population sizes. The heuristic reason is
straightforward. When $\tau_{I}/\tau_{E}$ is small, the spiking
pattern is homogeneous, and the distribution of spike counts in each time
window concentrates at a few small numbers. Larger $\tau_{I}/\tau_{E}$
means larger MFE sizes, which makes spike counts in each time window
more diverse. But when $\tau_{I}/\tau_{E}$ is too large and such that
all MFEs are large, such diversity shows slight decreases. In other
words, our simulation shows that the partial synchronization, rather
than the full synchronization, helps a neuronal network to produce
more information. 

\begin{figure}[!h]
\centerline{\includegraphics[width = 0.6\linewidth]{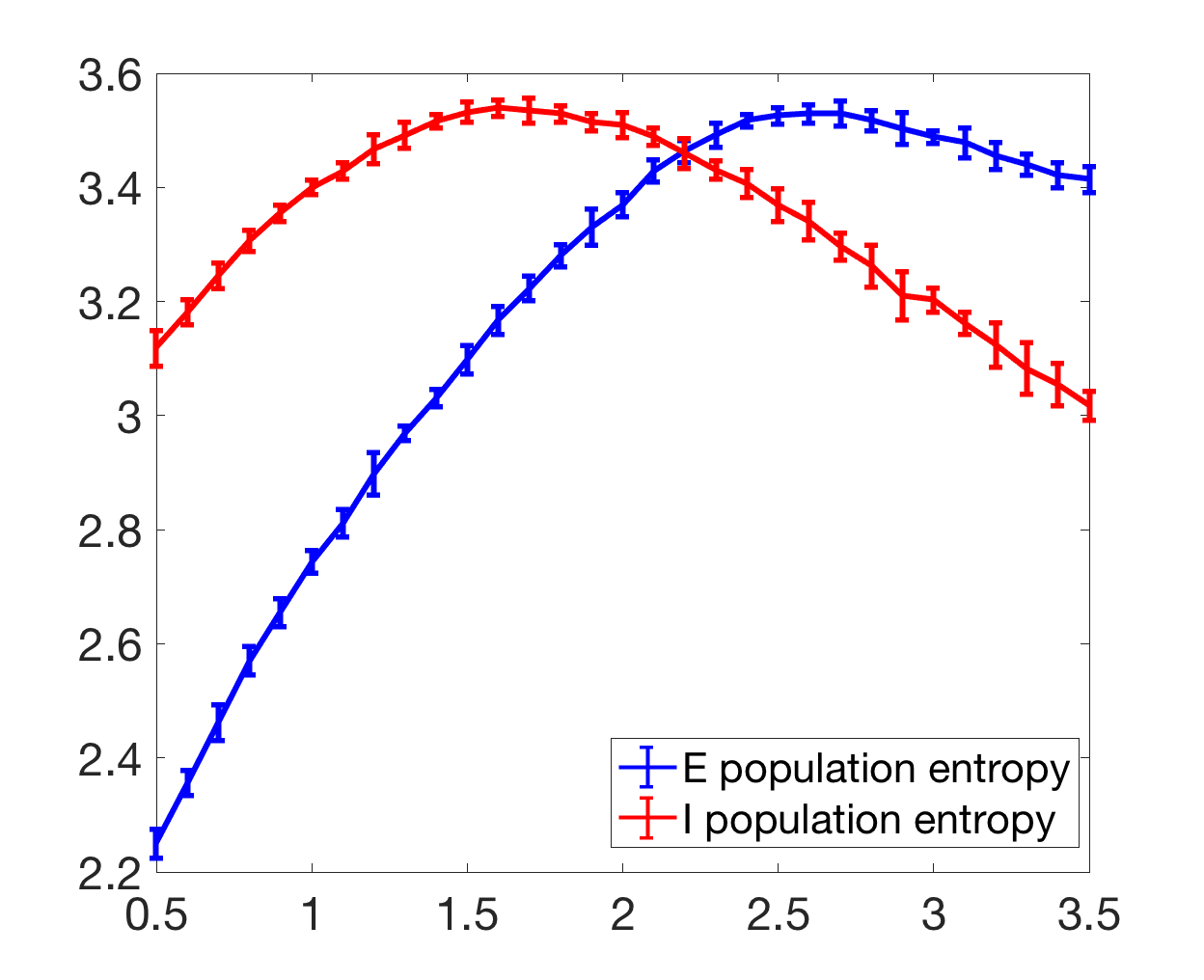}}
\caption{Entropy versus delay time for Model I. Each simulation runs up to $\mathcal{T} = 50$. }
\label{fig8}
\end{figure}

\section{Mutual information}
\subsection{Mutual information and conditional mutual information}
Based on the framework of spike train space defined in Section 3.1,
one can also define the mutual information. Let $T > 0$ be a fixed
time window size. A {\it coarse-grained
function $\mathcal{SPT} \rightarrow \mathbf{D}$ with respect
population set $C \subset \{1, \cdots, K\}$} is denoted by $\zeta_{C}$
if it gives joint spike count distributions from populations that belong to
$C$. Here $C = \{C_{1}, \cdots, C_{|C|}\}$ is a generic subset of $\{1, \cdots, K\}$. We have
\begin{equation}
\label{eq5-1}
  \zeta_{C}(S_{[0, T]}) = \sum_{n = 1}^{|C|} d^{n-1}\left (\Theta \left ( \sum_{j = 1}^{Z} \mathbf{1}_{\{
    \mathrm{Label}_{\xi_{j}}(1) \in C_{n}\}} \right ) \right )\,,
\end{equation}
where $\mathrm{Label}_{\xi_{j}}(1)$ means the first entry (index of
local population) of $j$-th spiked neuron, and $\Theta:
\mathbb{Z}_{\geq 0} \rightarrow \{1, \cdots, d\}$ is a partition
function that maps a spike count to an integer. Here the definition of
$\zeta_{C}$ can also be extended to the case of a particular type 
of neurons, or the case of spike counts with delays. 

Now consider two disjoint sets of
local populations $A, B \subset \{1, \cdots, K\}$. The {\it mutual
  information} is an information theoretical measure that is given
by
$$
  \mathrm{MI}_{T, \zeta}(A:B) = H_{T, \zeta_{A}} + H_{T, \zeta_{B}} - H_{T,
    \zeta_{A \cup B}} \,.
$$
In other words, the mutual information measures the information shared
by populations in $A$ and populations in $B$ when the spike count is
measured by the coarse-grained map $\xi$. 

The coarse-grained mutual information $\mathrm{MI}_{T, \zeta}(A:B)$ is
the mutual information between two spike trains produced by local
populations $A$ and $B$ that are processed by coarse-grained function
$\zeta_{A}$ and $\zeta_{B}$ respectively. By the data processing
inequality \cite{cover2012elements}, $\mathrm{MI}_{T, \zeta}(A:B)$ is
no greater than the actual mutual information between two unprocessed
spike trains on $A$ and $B$ respectively. Therefore, the ``true''
mutual information is larger than our measurement by using
coarse-graining functions.

\subsection{Mutual information in visual cortex models}
The following simulations aim to use mutual information to study the
correlation in our visual cortex models. We present the following $4$
results. 

{\bf Simulation I: Role of feedforward and feedback.} In
this study we first consider Model I, which is a the feedforward-feedback
network with only two local populations. We let the time window size
be $10$ ms with word length $= 1$. The partition function $\Theta$ maps $\mathbb{Z}_{+}$ to $\{0, 1,
\cdots, 10\}$ such that $\Theta(n) = i$ for $5i \leq n < 5(i+1)$ if $i
= 0, 1, \cdots, 9$ and $\Theta(n) = 10$ if $n \geq 50$. The mutual
information between the feedforward layer and the feedback layer
versus the coupling strength is plotted
in Figure \ref{fig9} Left. We consider three cases: (i) $\rho_{f} =
\rho_{b} = \rho$, (ii) $\rho_{f} = \rho, \rho_{b} = 0$, and (iii)
$\rho_{f} = 0, \rho_{b} = \rho$. One can see a clear increase of mutual
information between two layers with increasing coupling. And the
presence of both feedforward and feedback connection can significantly
increase the mutual information between two layers. This is
somehow expected: communications between neuronal networks can
increase the information shared between them. The same simulation is
done in Model II, which has many hypercolumns and geometric
structures. The mutual information is measured between $(2,2)$
hypercolumn of layer 4 and $(2,2)$ hypercolumn of layer 6. The corresponding partition function for Model II is $\Theta(n) = i$ for $10i \leq n < 10(i+1)$ if $i
= 0, 1, \cdots, 9$ and $\Theta(n) = 10$ if $n \geq 100$.  All other
parameters are same as those in Model I. The result is plotted in
Figure \ref{fig9} Right. We can see the same pattern as in Figure
\ref{fig9} Left. Note that when $\rho_{f} = 0, \rho_{b} = \rho$ and
$\rho$ is very large, the network becomes almost fully
synchronized. The mutual information slightly decreases in this regime
because the entropy decreases in a very synchronized network,
as seen in Figure \ref{fig8}. 

\vskip 1cm

\begin{figure}[!h]
    \centering
    \includegraphics[scale=0.18]{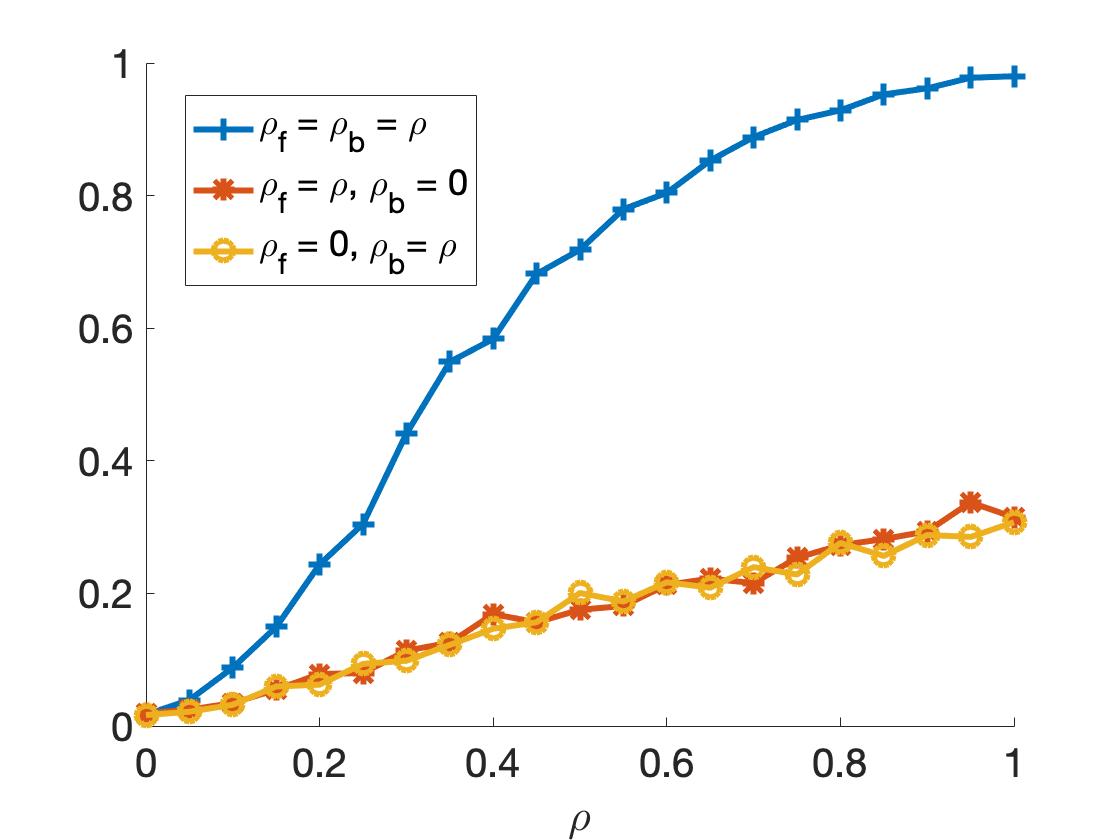}
    \includegraphics[scale=0.18]{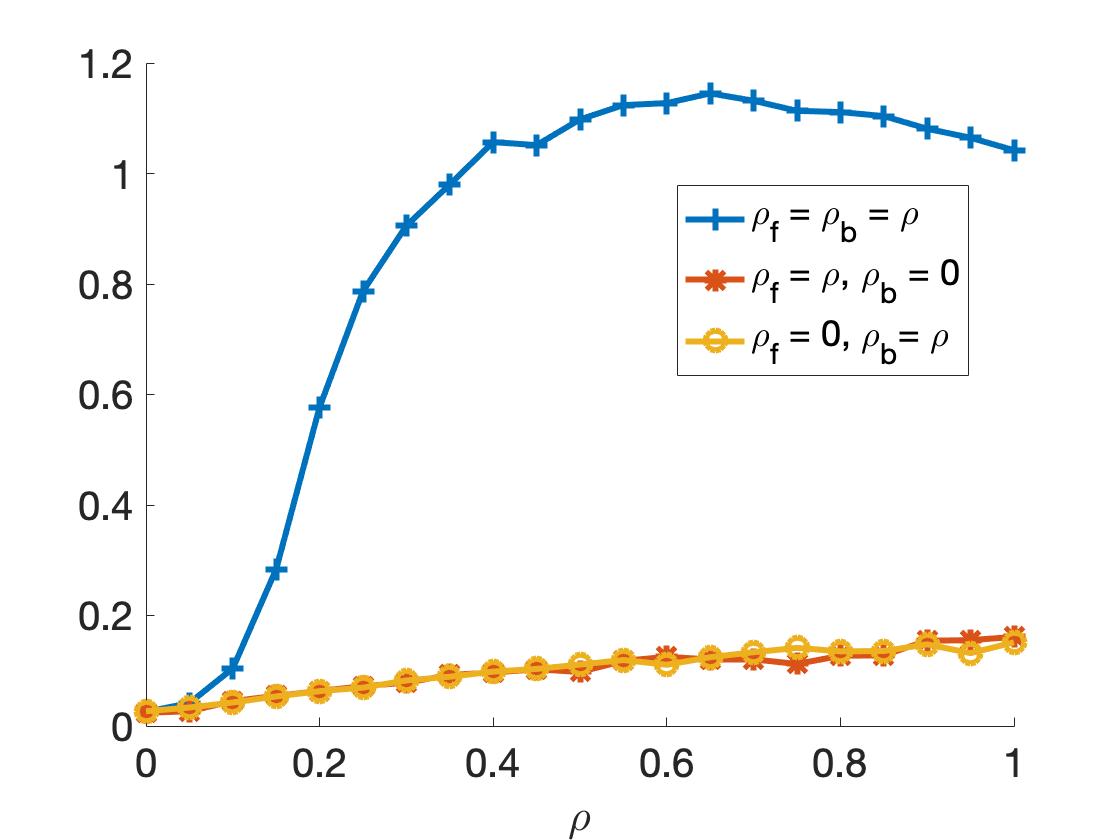}
    \caption{Mutual information versus coupling strength for Model I
      (left) and Model II (right). }
    \label{fig9}
\end{figure}

\vskip 1cm

{\bf Simulation II: Mutual information versus distance.} The
second simulation considers the dependence of mutual information on
the distance between two hypercolumns. We fix the coupling strength as
$\rho_{f} = \rho_{b} = 0.6$ and study mutual information between
hypercolumns for model II and model III. In Figure \ref{fig10} Top,
mutual information between $(1,1)$ hypercolumn and all other
hypercolumns are demonstrated for model II. Model III has orientational
columns. In Figure \ref{fig10} Bottom, we show the mutual information
between the vertical-preferred orientation column in $(1,1)$
hypercolumn and all other orientation columns in model III. The
logarithmic scale is used because the entropy at the bottom left corner is much larger than all
mutual information. We can
clearly see a decrease of mutual information with increasing
distance. This is consistent with our results in \cite{li2019stochastic} that
correlation of MFE sizes decays quickly with the distance. The fast
decline of mutual information is also supported by experimental
evidence. It is known
that MFEs are responsible for the Gamma rhythm in the cortex, which is
known to be local in many scenarios \cite{goddard2012gamma, lee2003synchronous,
  menon1996spatio}. In addition, in Figure \ref{fig10} Bottom, we can also see that
orientation columns with the same preferred orientations share higher
mutual information than those with opposite preferred orientations.

\vskip 1cm

\begin{figure}[!h]
    \centering
    \includegraphics[width=\linewidth]{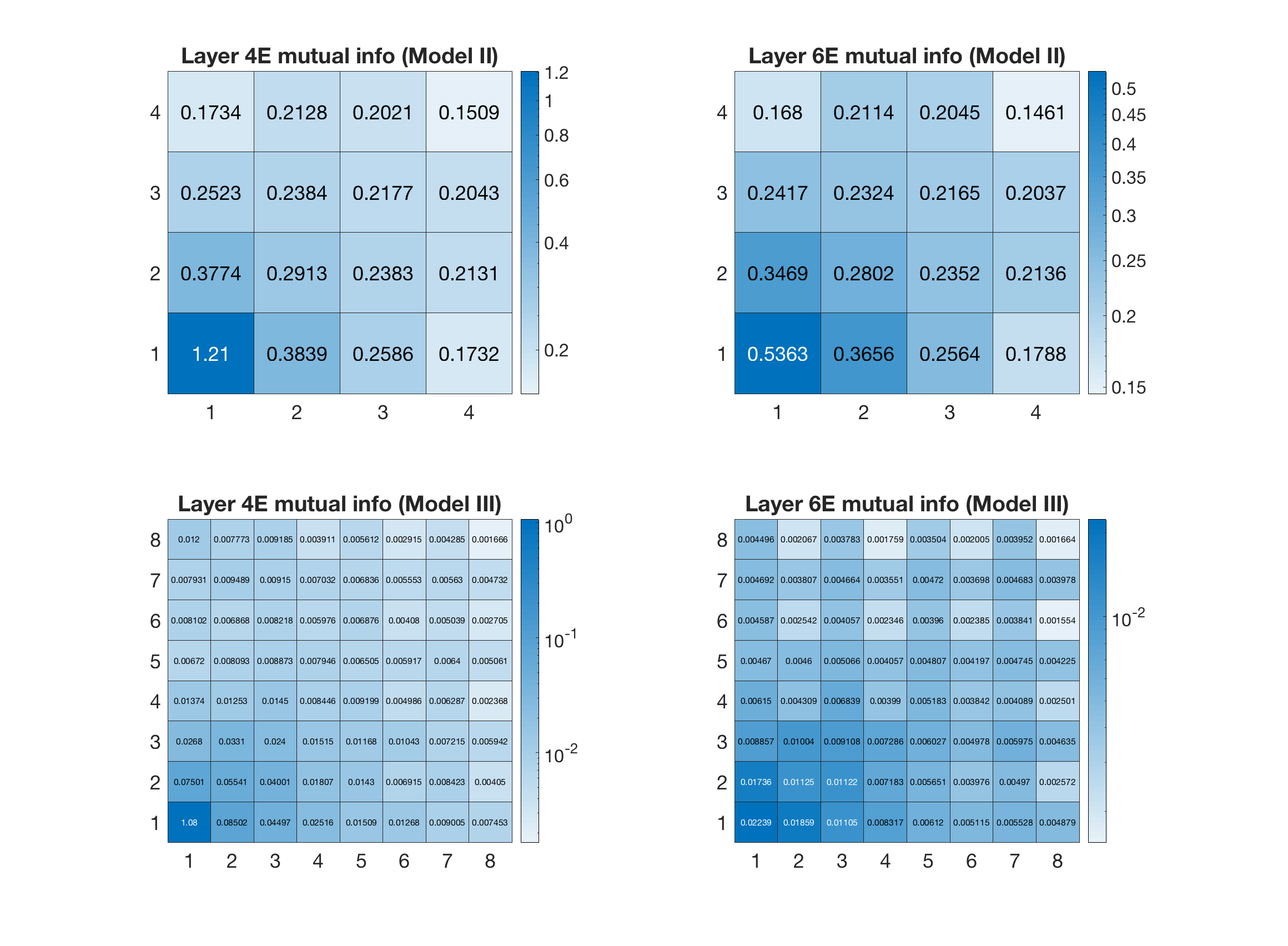}
    \caption{Top: heat map of mutual information with $(1,1)$
hypercolumn in Model II. Top left: Layer 4. Top right: Layer 6. Bottom: heat map
of mutual information with vertical-preferred orientation column in $(1,1)$
hypercolumn in Model III. Bottom left: Layer 4. Bottom right: Layer 6.}
    \label{fig10}
\end{figure}

\vskip 1cm
\section{Systematic measures: degeneracy and complexity}
\subsection{Definitions and rigorous results.}
As introduced in the introduction, systematic measures, including degeneracy, complexity, redundancy, and
robustness, are first proposed in the study of systems biology. When a
biological network is too large to be investigated in full detail,
systematic measures are used to describe the global characteristics of
the network. In \cite{tononi1999measures, li2012quantification}, degeneracy and complexity are quantified
as linear combinations of mutual information. This section uses the
idea of coarse-grained entropy to define two
systematic measures, i.e.,  degeneracy and complexity, for a spiking
neuronal network. The case of other systematic measures like the
redundancy \cite{tononi1999measures} can be studied analogously.

The definition of degeneracy and complexity relies on multivariate
mutual information. Still consider a neuronal network with $K$ local
populations. Let $\zeta_{I}$ be a coarse-grained function with
respect to $I \subset \{1, \cdots, K\}$. For three sets $A, B, C \subset \{1,
\cdots, K\}$, the {\it  multivariate mutual information} $MI(A:B:C)$ is given
as
\begin{eqnarray}
\label{MMI}
&&\\\nonumber
 MI_{T, \zeta}(A:B:C)& = & MI_{T, \zeta}(A:C) + MI_{T, \zeta}(B:C) -
                           MI_{T, \zeta}(A\cup B:C)\\\nonumber
&=& H_{T, \zeta_{A}} + H_{T, \zeta_{B}} +  H_{T, \zeta_{C}} -
    H_{T,\zeta_{A \cup B}} - H_{T,\zeta_{B \cup C}}  - H_{T,\zeta_{A
    \cup C}} + H_{T,\zeta_{A \cup B \cup C}}  \nonumber\,.
\end{eqnarray}

The following proposition follows immediately from the definition of
multivariate mutual information.

\begin{prop}
\label{limitcases}
\begin{itemize}
  \item[(a)] If $\zeta_{A}(S_{[0, T)})$, $\zeta_{B}(S_{[0, T)})$, and $\zeta_{C}(S_{[0,
    T)})$ are identical random variables, then $MI_{T, \zeta}(A:B:C) =   H_{T, \zeta_{A}}$. 
\item[(b)] If  $\zeta_{A}(S_{[0, T)})$, $\zeta_{B}(S_{[0, T)})$ are independent
  from $\zeta_{C}(S_{[0, T)})$, then $MI_{T, \zeta}(A:B:C) = 0$. 
\end{itemize}
\end{prop}
\begin{proof}
If all three random variables are equal, we have 
$$
H_{T,\zeta_{A
    \cup B}}  =  H_{T, \zeta_{A}}
$$
because the histogram of $\zeta_{A\cup B}(S_{[0, T)})$ is supported by the
diagonal set, which equals $\zeta_{A}(S_{[0, T)})$. Similar argument
implies that 
$$
  H_{T,\zeta_{B \cup C}}  = H_{T,\zeta_{A
    \cup C}} = H_{T,\zeta_{A \cup B \cup C}} = H_{T, \zeta_{A}} \,.
$$
This implies (a). 

If $\zeta_{C}(S_{[0, T)})$ is independent of $\zeta_{A}(S_{[0, T)}),
\zeta_{B}(S_{[0, T)})$, we have
$$
  H_{T,\zeta_{B \cup C}} = H_{T,\zeta_{B}} + H_{T,\zeta_{C}}, \quad
  H_{T,\zeta_{A \cup C}} = H_{T,\zeta_{A}} + H_{T,\zeta_{C}}
$$
and
$$
  H_{T,\zeta_{A \cup B \cup C}}  = H_{T,\zeta_{A \cup B}}  +
  H_{T,\zeta_{C}}  \,.
$$
It follows from equation \eqref{MMI} that 
$$
  MI_{T, \zeta}(A:B:C) = MI_{T, \zeta}(A:B) - MI_{T, \zeta}(A:B) = 0 \,.
$$
\end{proof}

Now consider a neuronal network with $K$ local populations. Let two
disjoint sets $\mathcal{I}, \mathcal{O} \subset \{1, \cdots, K\}$ denote the {\it
  input set} and {\it output set} of this network. The {\it degeneracy}
$\mathcal{D}_{T, \zeta}( \mathcal{I}: \mathcal{O})$ is given by the
weighted sum of multivariate mutual informations:
\begin{equation}
\label{deg}
\mathcal{D}_{T, \zeta}( \mathcal{I}: \mathcal{O}) = \sum_{0 \leq k
  \leq | \mathcal{I}|} \frac{1}{2 {|
    \mathcal{I}| \choose k}}MI_{T,
  \zeta}(\mathcal{I}_{k}: \mathcal{I}_{k}^{c}: \mathcal{O}) \,,
\end{equation}
where the summation goes through all possible partitions
of $\mathcal{I}$. Set $\mathcal{I}_{k}$ means a subset of
$\mathcal{I}$ with $k$ local populations. Degeneracy measures how much more information different
components of the input set share with the output set than expected if
all components are independent. Degeneracy is high if many structurally different
components in the input set can perform similar functions on a
designated output set. A neuronal network is said to be {\it degenerate} if
$\mathcal{D}_{T, \zeta}(\mathcal{I}: \mathcal{O}) > 0$ for some choice
of $T, \zeta, \mathcal{I}$, and $\mathcal{O}$. 

The (structural) {\it complexity} $\mathcal{C}_{T, \zeta}( \mathcal{I}:
\mathcal{O})$ is given by the weighted sum of mutual information
between components of the input set. 
\begin{equation}
\label{com}
\mathcal{C}_{T, \zeta}( \mathcal{I}: \mathcal{O}) = \sum_{0 \leq k
  \leq | \mathcal{I}|} \frac{1}{2 {|
    \mathcal{I}| \choose k}}MI_{T,
  \zeta}(\mathcal{I}_{k}: \mathcal{I}_{k}^{c}) \,,
\end{equation}
The complexity measures how much codependency in a network appears
among different components of the input set. Again, a neuronal network
is said to be (structurally) {\it complex} if
$\mathcal{C}_{T, \zeta}(\mathcal{I}: \mathcal{O}) > 0$ for some choice
of $T, \zeta, \mathcal{I}$, and $\mathcal{O}$. 

The degeneracy at two
limit cases can be given by proposition \ref{limitcases} easily. When
the input is independent of the output, we have zero degeneracy. When
a neuronal network is fully synchronized, and $A, B, C$ are three
local populations with $\mathcal{I} = \{A, B\}, \mathcal{O} = C$, then
the degeneracy equals to $H_{\zeta, T}(A)$, which is positive. Our
numerical simulation result will confirm this.

Finally, we have the following lemma regarding the connection between
degeneracy and complexity. 
\begin{lem}
\label{lem51}
For any choice of $T, \zeta, \mathcal{I}$, and $\mathcal{O}$ and any
decomposition $\mathcal{I} = \mathcal{I}_{k} + \mathcal{I}_{k}^{c}$,
we have
$$
  MI_{T, \zeta}( \mathcal{I}_{k} : \mathcal{I}_{k}^{c}: \mathcal{O})
  \leq \min \{ MI_{T, \zeta}(\mathcal{I}_{k} : \mathcal{I}_{k}^{c}),
  MI_{T, \zeta}(\mathcal{I}_{k} : \mathcal{O}),  MI_{T,
    \zeta}(\mathcal{I}_{k}^{c}: \mathcal{O}) \} \,. 
$$
\end{lem}
\begin{proof}
This lemma is a discrete version of Lemma 5.1 of \cite{li2016systematic}. We
include it for the sake of completeness of this paper. It is
sufficient to prove that for any three discrete random variables $X$,
$Y$, and $Z$ with joint probability distribution function $p(x,y,z) =
\mathbb{P}[ X = x, Y = y, Z = z]$,
$$
  MI(X:Y:Z) \leq \min \{ MI(X:Y), MI(X:Z), MI(Y:Z)\} \,.
$$
It follows from the definition of the multivariate mutual information
and some elementary calculations that 
\begin{eqnarray*}
  MI(X:Y:Z) &=& H(X) + H(Y) + H(Z) - H(X,Y) - H(X,Z) - H(Y,Z) + H(X,Y,Z) \\
&=& MI(X:Y) - (H(X,Z) + H(Y,Z) - H(Z) - H(X,Y,Z)\\
&=&MI(X:Y) - MI(X:Y \,|\, Z) \,,
\end{eqnarray*}
where the latter time is the conditional mutual information. The
conditional mutual information is nonnegative, as a direct corollary
of Kullback's inequality \cite{yeung2012first}. Hence 
$$
  MI(X:Y:Z) \leq MI(X:Y) \,.
$$
Inequalities $MI(X:Y:Z) \leq MI(X:Z)$ and $MI(X:Y:Z) \leq MI(Y:Z)$
follow analogously. This completes the proof.
\end{proof}

The following theorem is a straightforward corollary of Lemma
\ref{lem51}. 

\begin{thm}
\label{degcom}
For any choice of $T, \zeta, \mathcal{I}$, and
$\mathcal{O}$, 

$$
  \mathcal{D}_{T, \zeta}( \mathcal{I}: \mathcal{O}) \leq \mathcal{C}_{T, \zeta}( \mathcal{I}: \mathcal{O}) 
  \,.
$$
\end{thm}

Heuristically, Theorem \ref{degcom} says that a neuronal network with
high degeneracy must has high (structural) complexity as well.

\subsection{degeneracy and complexity in a visual cortex model}
We consider the following three numerical simulations to measure the degeneracy and
complexity in model II and model III. Because the estimation of entropy
is not satisfactory when the simulation is under-sampled, we limit the
cardinality of $| \mathcal{I}|$ to three and only consider the case of
$| \mathcal{O} | = 1$

{\bf Numerical simulation I. } In Model II, we consider the input set
$\mathcal{I} = \{L_{6}, L_{7}, L_{10}\}$ and $\mathcal{O} =
L_{22}$. In other words the input set consists of hypercolumns $(2,2),
(2,3)$, and $(3,2)$ in layer 4 and the output set is local population
$(2,2)$ in layer 6. The synapse delay times are chosen to be $\tau_{E}
= 2$ ms and $\tau_{I} = 4.5$ ms. The time window size is $T = 5$ ms. The
coarse-grained function $\zeta_{C}$ with respect to a generic set $C$
is given as in equation \eqref{eq5-1}, with a partition function $\Theta$ maps $\mathbb{Z}_{+}$ to $\{0, 1, 2, 3, 4,
5\}$ such that $\Theta(n) = i$ for $20 i \leq n < 20 (i+1)$ if $i
= 0, 1, \cdots, 4$ and $\Theta(n) = 5$ if $n \geq 100$.  Figure
\ref{fig11} Left shows the dependence of degeneracy $\mathcal{D}_{\zeta,
  T}( \mathcal{I}, \mathcal{O})$ versus the coupling strength between
layers $\rho_{f} = \rho_{b} = \rho$. We can see that the degeneracy
and complexity increases with $\rho$ in general. When $\rho$ is too
large, the recurrent excitation cannot be tempered by inhibitions and
the network fires at a very high rate that can not be captured by the
partition map $\Theta$. As a result, both $\mathcal{D}$ and
$\mathcal{C}$ drops to very small values.

\medskip

{\bf Numerical simulation II. } Now we take orientation columns into
considerations. In Model III, we let the input set $\mathcal{I}$
consist of three orientation columns in hypercolumn $(2,2)$ of layer
4, with orientation preferences $0$ deg, $45$ deg, and $90$ deg respectively. The
output set is the orientation column in hypercolumn $(2,2)$ of layer
6, with an orientation preference $0$ deg. The synapse delay times are
chosen to be $\tau_{E} = 2$ ms and $\tau_{I} = 4.5$ ms. The time
window size is $T = 5$ ms. The coarse-grained function $\zeta_{C}$ is
similar as in {\bf Numerical simulation I}, except the partition
function takes value $\Theta(n) = i$ for $10 i \leq n < 10 (i+1)$ if $i
= 0, 1, \cdots, 4$ and $\Theta(n) = 5$ if $n \geq 50$, because an
orientation column contains less neurons than a hypercolumn. The
strength of long range connection is given by $C_{long} = 0.5$. Figure
\ref{fig11} Middle shows the dependence of degeneracy $\mathcal{D}_{\zeta,
  T}( \mathcal{I}, \mathcal{O})$ versus the coupling strength between
layers $\rho_{f} = \rho_{b} = \rho$. Again, stronger coupling between
layers leads to larger degeneracy, until the network firing rate is
too high to be captured by the given partition function. 

\medskip

{\bf Numerical simulation III. } The third simulation studies the
effect of long range connections. Still in Model III, we let the input set $\mathcal{I}$
consist of three orientation columns with orientation preference $0$
deg in hypercolumns $(2,2), (2,3)$, and $(3,2)$ of layer 4. The output set is the orientation column in hypercolumn $(2,2)$ of layer
6, with an orientation preference $0$ deg. Parameters like the synapse
delay times, the time window size, the coarse-grained function
$\zeta_{C}$ are identical to those of {\bf Numerical simulation
  II}. The feedforward and feedback strengths are $\rho_{b} = \rho_{f}
= 0.6$. And the strength of long range connections is the main control
parameter. Figure
\ref{fig11} shows the dependence of degeneracy $\mathcal{D}_{\zeta,
  T}( \mathcal{I}, \mathcal{O})$ versus the strength of long range
connections $C_{long}$. We can see that a stronger long range coupling
between orientation columns also leads to a higher degeneracy. 

\begin{figure}[!h]
\centerline{\includegraphics[width = \linewidth]{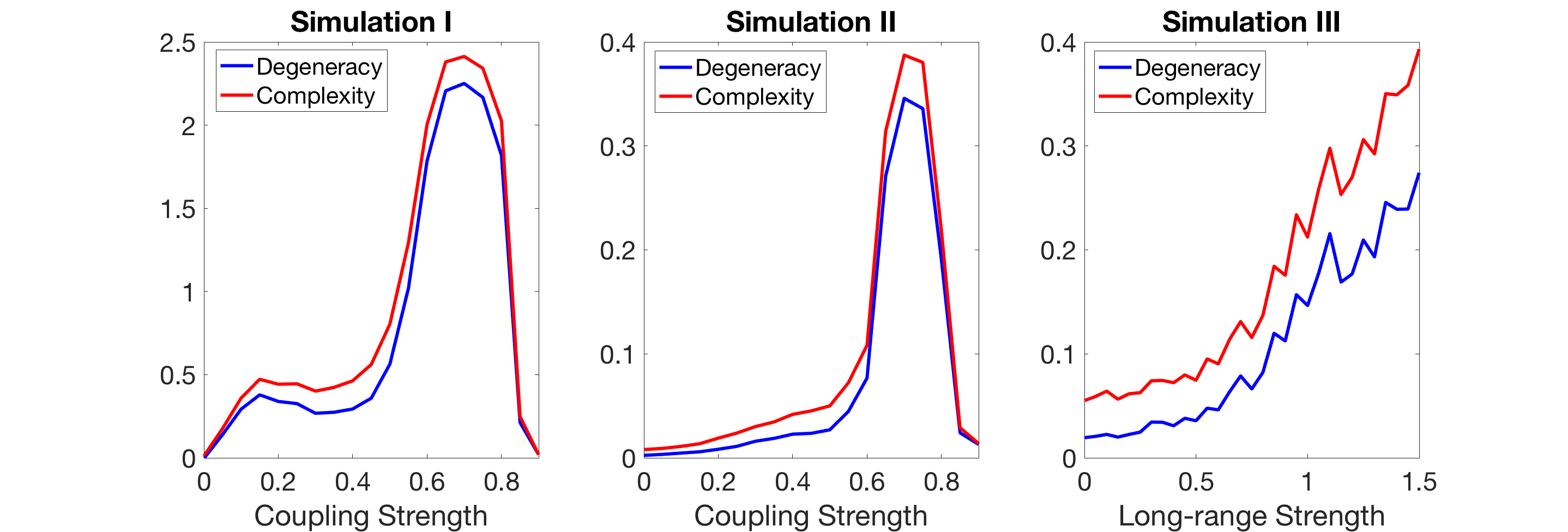}}
\caption{Degeneracy and complexity in model II and III. Left, middle,
  and right panels are from from {\bf Numerical
  simulation I - III} respectively. Length of simulation = 20 with $16$
independent threads. }
\label{fig11}
\end{figure}

\section{Conclusion}
This paper investigates a few information-theoretic measures of a class
of structured neuronal networks. Neurons in the network are of the
integrate-and-fire type. Being consistent with our earlier papers
\cite{li2019well, li2019stochastic}, membrane potentials are set to be discrete to make the
model mathematically and computationally simple. The network consists
of many local populations, each of which has its own external drive
rate. Biologically, one local population could be a hypercolumn in the
cortex or an orientational column in the visual cortex. We provided
several different network models for the purpose of examining
information-theoretic measures. The most complicated model (Model III)
aims to model two layers (layer 4 and layer 6) of the primary visual cortex, each layer has
$4\times 4$ hypercolumns, each hypercolumn has $4$ orientational
columns. 

Then we use the idea of coarse-graining to define the coarse-grained
entropy. The motivation is that the naive definition of neuronal
entropy works poorly for large neuronal networks. One needs
unrealistically large sample to estimate the entropy in a trustable
way. In addition it is known that many neuronal network can generate
Gamma-type rhythms. Hence when calculating the entropy, we chose to ignore the precise order of
neuronal spikes from the same local population if they fall into the
same small time bin. This entropy mainly captures the uncertain in the
oscillations produced by the neuronal network. One can easily compute
the entropy for large scale neuronal networks. The mutual information
can be defined analogously. 

The coarse-grained entropy and the mutual information is examined
through our examples. We found that the coarse-grained entropy mainly
capture the information contained in the rhythm produced by a local
population. Under suitable setting of the partition function, the
coarse-grained entropy reaches maximal value when the partial
synchronization has the most diversity, and decreases when the spiking pattern
is either homogeneous or fully synchronized. Furthermore, in our
two-layer network model, we found that stronger connections between
layers can produce higher mutual information between layers. This is
very heuristic, as stronger coupling between layers makes the firing
patterns from two layers more synchronized. 

In the end, we attempts to quantify two systematic measures, namely
the degeneracy and the complexity, for spiking neuronal network
models. These systematic measures are originally proposed in the study
of systems biology. They can be written as linear combinations of
mutual informations. Therefore, after defining coarse-grained entropy
and mutual information, these systematic measures can be defined
analogously. We found that the inequality proved in our earlier paper
\cite{li2016systematic} still holds, which says that the degeneracy is always
smaller than the complexity, or a system with high degeneracy must be
structurally complex. Finally, we numerically computed degeneracy and
complexity for our two-layer cortex models. We found that at certain
range of parameters, stronger coupling between layers, as well as
stronger long-range connectivities, leads to both larger degeneracy and
larger complexity.

\bibliography{myref}
\bibliographystyle{amsplain}
\end{document}